\documentclass[11pt,reqno]{amsart}
\usepackage{natbib}
\usepackage{graphics,graphicx,amsmath,amssymb,epsfig,epstopdf,stmaryrd,color,flafter,xspace,lscape,threeparttable,pdfpages}
\usepackage{csquotes,amsfonts,amsthm,geometry,array,booktabs,latexsym,caption,subfigure,subfig,multirow,tabularx,rotating,bigstrut}
\usepackage[font=small,labelfont=bf]{caption}
\DeclareCaptionFont{tiny}{\tiny}
\everymath{\displaystyle}

\RequirePackage{setspace,indentfirst,epsfig,psfrag,ifthen,ifpdf}
\usepackage{changepage}

\newtheorem{theorem}{Theorem}
\theoremstyle{plain}

\newtheorem{assumption}{Assumption}
\newtheorem{definition}{Definition}
\newtheorem{corollary}{Corollary}

\newtheorem{subassumption}{Assumption\normalfont}[assumption]

\newenvironment{assumption*}
 {\ifnum\value{subassumption}=0 \stepcounter{assumption}\fi\subassumption}
 {\endsubassumption}
\newenvironment{assumption+}[1]
 {\renewcommand{\thesubassumption}{#1}\subassumption}
 {\endsubassumption}

\newtheorem{subdefinition}{Definition\normalfont}[definition]

\newenvironment{definition*}
 {\ifnum\value{subdefinition}=0 \stepcounter{definition}\fi\subdefinition}
 {\endsubdefinition}
\newenvironment{definition+}[1]
 {\renewcommand{\thesubdefinition}{#1}\subdefinition}
 {\endsubdefinition}

\newtheorem{subtheorem}{Theorem\normalfont}[theorem]

\newenvironment{theorem*}
 {\ifnum\value{subtheorem}=0 \stepcounter{theorem}\fi\subdefinition}
 {\endsubtheorem}
\newenvironment{theorem+}[1]
 {\renewcommand{\thesubtheorem}{#1}\subtheorem}
 {\endsubtheorem}

\newtheorem{subcorollary}{Corollary\normalfont}[corollary]

\newenvironment{corollary*}
 {\ifnum\value{subcorollary}=0 \stepcounter{corollary}\fi\subcorollary}
 {\endsubcorollary}
\newenvironment{corollary+}[1]
 {\renewcommand{\thesubcorollary}{#1}\subcorollary}
 {\endsubcorollary}

\newcommand\blfootnote[1]{%
  \begingroup
  \renewcommand\thefootnote{}\footnote{#1}%
  \addtocounter{footnote}{-1}%
  \endgroup
}

\theoremstyle{definition}
\newtheorem{assump}{Assumption}

\newtheorem{lemma}{Lemma}

\newtheorem{proposition}{Proposition}

\newtheorem*{remark*}{Remark}

\numberwithin{equation}{section}

\begin{document}

\title[Role Models in Roy models]{Role models and revealed gender-specific costs of STEM in an extended Roy model of major choice.}
\author{Marc Henry}\address{The Pennsylvania State University} 
\author{Romuald M\'eango}\address{Oxford University}
\author{Isma\^^22el Mourifi\'e}\address{University of Toronto, Washington University in St. Louis and NBER}\noindent \blfootnote{\scriptsize{The first version is of August 21, 2019. The present version is of July 14, 2023. This research was supported by SSHRC Grant 435-2018-1273 and Leibniz Association Grant SAW-2012-ifo-3. The research was conducted in part, while Isma\^^22el Mourifi\'e was visiting the Becker Friedman Institute at the University of Chicago. He thanks his hosts for their hospitality and support. The authors thank the editor Elie Tamer, the associate editor and a referee for very detailed and insightful reports, that were instrumental to significant improvements in the paper (the usual disclaimer applies of course). The authors also acknowledge the exceptional research assistance of Ivan Sidorov, as well as helpful comments from Roy Allen, Samuel Cohen, Chris Dobronyi, Jonathan Eaton, Jim Heckman, Hiro Kasahara, Thomas Russell, Azeem Shaikh and seminar audiences at Cambridge, Chicago, Notre Dame, Rice, Texas A\&M, Vanderbilt and York. Correspondence address: Department of Economics, Max Gluskin House, University of Toronto, 150 St. George St., Toronto, Ontario M5S 3G7, Canada}}

{\scriptsize 

\begin{abstract}
We derive sharp bounds on the non consumption utility component in an extended Roy model of sector selection. We interpret this non consumption utility component as a compensating wage differential. The bounds are derived under the assumption that potential utilities in each sector are (jointly) stochastically monotone with respect to an observed selection shifter. 
The research is motivated by the analysis of women's choice of university major, their under representation in mathematics intensive fields, and the impact of role models on choices and outcomes.
To illustrate our methodology, we investigate the cost of STEM fields with data from a German graduate survey, and using the mother's education level and the proportion of women on the STEM faculty at the time of major choice as selection shifters.
\end{abstract}

}

\maketitle
\thispagestyle{empty}

{\scriptsize \textbf{Keywords}: Roy model, partial
identification, stochastic monotonicity, sharp bounds, sharp testable implications, education sector choice, role models, women in STEM.

}

%%%%%%%%%%%%%%%

%    INTRODUCTION

%%%%%%%%%%%%%%%

\section*{introduction}

A large body of evidence supports sorting in the labor market as a dominant mechanism in the explanation of the residual gender wage gap. See \cite{Goldin:2021} and references therein. Two salient aspects of sorting in the labor market are the persistent under representation of women in Science, Technology, Engineering and Mathematics (hereafter STEM), which  is perceived to be one of the main drivers of the gender wage gap (see \cite{daymont1984}, \cite{zafar2013}, and \cite{SHB:2019}), and the documented influence of role models on sorting choices and eventual outcomes. A survey of the scholarship on both these aspects can be found in \cite{KG:2017}.

We propose to examine the issue of under representation of women in STEM from the point of view of college major choice, which is a major sorting mechanism for skilled labor, and with the lens of the Roy model.  However, the traditional version of the Roy model cannot accommodate stylized facts relating to major and occupational choice. The literature on major and occupational choice, reviewed in \cite{AAM:2016}, emphasizes the importance of factors beyond expected earnings. This motivates the current theoretical analysis of partial identification of gender-specific non pecuniary drivers of major choice within what is generally known in the literature as the extended Roy model (see \cite{HV:99} for details on model genealogy and attribution). The model allows for observed heterogeneity at the individual level, and unobserved heterogeneity at the educational sector level, to influence choices.
\cite{BKT:2011} and \cite{HM:2011} analyze the extended Roy model and provide competing strategies to identify the non pecuniary driver of choice (see the following subsection for details). We propose a utility-based version of the extended Roy model, that allows the implicit non pecuniary costs of each sector to also depend on potential outcomes, unlike \cite{BKT:2011} and \cite{HM:2011}. We make alternative partially identifying assumptions to those in \cite{BKT:2011} and \cite{HM:2011} and derive sharp bounds on utilities and on the non pecuniary driver of choice. The model allows for multiple interpretations of the non pecuniary component, including job amenities, anticipated lack of support for women in mathematics intensive education and occupations, as well as social conditioning and gender stereotyping of occupations.

We combine an extended Roy model of sorting with shape restrictions induced by the influence of role models.
The observed realized earnings variable~$Y$ is equal to~$Y=Y_1D+Y_0(1-D)$, where~$Y_0$ and~$Y_1$ are unobserved potential earnings in Sector~$0$ (non mathematics intensive college major, hereafter referred to as non-STEM) and Sector~$1$ (mathematics intensive college major, hereafter referred to as STEM), and~$D$ is the observed random sector selection indicator. Individuals choose Sector~$d=0,1,$ when $\mathbb E[u(Y_d,d,W)\vert \mathcal I]>\mathbb E[u(Y_{1-d},1-d,W)\vert \mathcal I]$, where~$u$ is the utility function, $\mathcal I$ is the individual's information set at the time of choice, and~$W$ is an~$\mathcal I$-measurable vector of individual and sector observed variables that influence choice and outcomes. This specification differs in two respects with the extended Roy models in the literature: first we do not impose a tie-breaking mechanism, which is important in case of interval censored observations; second, we do not impose quasi-linearity of utility, and allow varying marginal rates of substitution between consumption and sector amenities. The latter feature is shared with \cite{LP:2021}.

We formalize the influence of role models on choices and eventual outcomes with the assumption that a  subvector~$Z$ of the observable variables~$W:=(X,Z)$ can only affect potential expected utilities in one direction. Precisely, for all~$\nu$ and~$x$, the joint probability that~$\mathbb E[u(Y_1,1,W)\vert\mathcal I]>\nu$ and~$\mathbb E[u(Y_0,0,W)\vert\mathcal I]>\nu$, conditionally on~$X=x$ and~$Z=z$, is monotone non decreasing in~$z$. This shape restriction is inspired by monotone instrumental variables, and related to \cite{MP:2000}, \cite{BGIM:2007} and \cite{MHM:2020}.
We argue that factors that positively impact the formation of cognitive and non cognitive skills, while reducing perceived costs of the STEM sector, are likely to satisfy this shape restriction.
Maternal educational attainment and the proportion of female role models (faculty, alumni, invited speakers) are prime candidates (see \cite{BGME:2018} and \cite{RCM:2014}).

The objective of the analysis is to characterize the collection of utility functions that rationalize observed choices. We will refer to this set as the {\em identified set} (also known as sharp identified region in the partial identification literature). We also derive a closed form expression for the minimum perceived cost of mathematics intensive education and activity to rationalize the observed under representation of women in that sector. We interpret this cost as a wage differential that compensates for the real or perceived disamenities of the STEM sector for women. Both the identified set and the closed form lower bound are amenable to inference using existing methods for stochastic monotonicity tests, conditional moment inequalities and intersection bounds.

\cite{MHM:2020} document rejection of the Roy model of sorting on labor market outcomes for women in the sample of German graduates in the 2005 and 2009 graduating cohorts of the German DZHW Graduate Survey (see \cite{DZHW}). We therefore illustrate our theoretical results with an analysis of women's choice of major within the framework of our extended Roy model and confidence regions for the minimum cost of STEM that rationalizes choices. 
We assume that the proportion of women on the STEM faculty in the individual's region at the time of major choice (as a proxy for the presence of role models and better amenities for women) only affects potential labor market outcomes of women graduates positively. We find significant costs of choosing STEM fields for German women from the former Federal Republic, only when assuming perfect foresight (i.e., perfect knowledge of the potential utilities) when choosing the major. The costs are particularly pronounced for women in the lower quartile of the income distribution and for women whose region had low rates of feminization of the STEM faculty at the time of major choice.

% RELATED LITERATURE

\subsection*{Related Literature}

Early results on identification of a parametric version of the extended Roy model are given in \cite{Vijverberg:93}. Nonparametric identification in extended Roy models is analyzed in \cite{HM:2011} and \cite{BKT:2011}. \cite{BKT:2011} restrict attention to non pecuniary costs that are constant across agents, and \cite{HM:2011} to non pecuniary costs that depend on observables only, and not on potential outcomes. \cite{BKT:2011} identify the pecuniary cost function from a common lower bound on the supports of both potential incomes, or alternatively from the exclusion of a location shifter in the potential outcome equations. \cite{HM:2011} combine an additive decomposition of the agent's expected earnings with some smoothness and a continuous covariate.
\cite{LP:2021} also propose an extended Roy model in structural form. The selection equation in \cite{LP:2021} is the same as in our perfect foresight case, so that implicit non pecuniary costs can also depend on potential outcomes. However, identification is achieved from the variation in a utility shifter that is excluded from the potential outcome equations. The latter which may be hard to find and justify in some empirical contexts, such as the one we entertain here.

The inspiration for the shape restriction that drives partial identification goes back at least to \cite{MP:2000}. More recently, \cite{MHM:2020} derive sharp testable implications of stochastic monotonicity constraints in the traditional Roy model. In a related earlier contribution, \cite{Chetverikov:2013} shows that regression monotonicity is the testable implication of
the monotone treatment response assumption and monotone treatment selection assumption
introduced in \cite{MP:2000}. 

The empirical illustration of our proposed methodology fits into a growing applied literature
on the causes of under representation of women in university STEM majors. That literature is already sizable, as evidenced by the surveys in \cite{AAM:2016} and \cite{KG:2017}. The most closely related in terms of methodology is \cite{Saltiel:2018}, which looks at women's college major choice through the lens of a dynamic generalized Roy model inspired by \citeauthor{HHV:2018} [\citeyear{HHV:2016,HHV:2018}]. As concerns empirical findings, our paper relates more specifically to the literature on the role of female professors and gender specific amenities in shaping major choice and educational outcomes. This includes \cite{CR:95}, \cite{BGME:2018}, \cite{CPW:2010}, \cite{CM:2021} on the influence of female professors as role models, and \cite{daymont1984}, \cite{MP:2017}, \cite{KSW:2018}, \citeauthor{WZ:2015} [\citeyear{WZ:2015,WZ:2018}]  and \cite{MHM:2020}
on the importance of gender specific amenity preferences and gender specific response to non pecuniary factors in the valuation and choice of occupations.

%    OUTLINE

\subsection*{Outline} 
The next section presents the extended Roy model and the main identification results. Section~\ref{sec:frame} focuses on the special case of perfect foresight. Section~\ref{sec:inference} discusses inference methods. Section~\ref{sec:sim} presents a simulation study inspired by the under representation of women in STEM fields, and~Section~\ref{sec:appli} illustrates the methodology with an analysis of women's major choices in Germany. The last section concludes. Proofs of the main results are collected in the appendix.

%%%%%%%%%%%%%%%%%%%%%

%  IMPERFECT FORESIGHT

%%%%%%%%%%%%%%%%%%%%%

\section{Extended Roy model with role models}
\label{sec:if}

We adopt the framework of the potential
outcomes model $Y=Y_1D+Y_0(1-D).$ Observed outcome~$Y$ has support~$\mathcal Y\subseteq[\underline b,\infty)$, (with~$\underline b\in\mathbb R\cup\{-\infty\}$ and~$\underline b=0$ or~$\underline b=-\infty$ in most cases of interest),
$D$ is an observed selection indicator, which takes value $1$ if Sector~1 is chosen, and $0$ if Sector $0$ is chosen, and $Y_1$, $Y_0$, are
unobserved potential outcomes. In the context of major choice, the outcome of interest will be income in the year following graduation. Sector~1 will consist of all STEM majors and Sector~0, the rest. Decision makers choose their sector of activity based on the realizations of~$Y_0$ and~$Y_1$, and a vector of observed exogenous characteristics~$Z$ with support $\mathcal Z\subseteq\mathbb R^{d_z}$. The vector~$Z$ can be a vector~$Z=(Z_0,Z_1)$ of sector specific cost shifters\footnote{Unlike potential outcomes~$Y_d$, which are only observed in the chosen sector, the sector specific cost shifters are observed in both sectors.}.  In the context of women's major choice, for instance, this could be the proportion of women on the faculty in both sectors in the individual's region or prospective university at the time of choice. The shifter~$Z$ can also be a non sector specific (possibly vector) cost shifter. In the context of women's major choice, this could be the mother's education attainment. Additional observed exogenous covariates will be omitted from the notation. In the context of major choice in Germany, these include gender, visible minority status and a dummy for residence in the former East Germany, as a less affluent region. 

Since $Y,$ $D$ and $Z$ are observed, the distribution~$\pi$ of $(Y,D,Z)$ is directly identified from the data. We call~$\Pi$ the set of admissible data generating processes. Unless otherwise specified, $\Pi$ is the set of probability distributions on $\mathcal Y\times\{0,1\}\times\mathcal Z$.
We summarize the model with the following assumptions.
\begin{assumption}[Potential outcomes]
\label{ass:PO}
Observed outcomes are the realizations of a random variable~$Y$ satisfying $Y=Y_1D+Y_0(1-D)$, where $(Y_0,Y_1)$ is a pair of possibly dependent unobserved random variables with support~$\mathcal Y\subseteq[\underline b,\infty)$, and $D$ is an observed indicator variable.
\end{assumption}

The original Roy model posits sector selection based only on the comparison of potential outcomes, so that~$Y=\max\{Y_0,Y_1\}$. Given our focus on the under representation of women in STEM fields, and rejections of the original Roy model selection rule in~\cite{MHM:2020}, we entertain the possibility that other factors affect sector selection in favor of the non STEM sector. In the context of women's major choices, there might be a gender-specific cost of studying in a STEM field or a gender-specific cost of working in the STEM sector. The former may be the result of the lack of support for female students, or a fear of mathematics carried over from schooling (see \cite{XFS:2015} for a survey of the sociological literature on the subject). The latter may be related to family-friendliness of employment outside STEM.

Hence, in our model, the decision by the individual is based on the comparison between expected utilities~$\mathbb E[u(Y_d,d,Z)\vert \mathcal I]$ for~$d\in\{0,1\}$, where $Y_d$ is the potential outcome (wage) in Sector~$d$ and $Z$ is the (possibly vector) value of the non pecuniary factor that may also affect utility. 
Choices are made using expectations based on the decision maker's information set~$\mathcal I$ at the time of decision. We assume throughout this section that the instrument~$Z$ is~$\mathcal I$-measurable, since it is a vector of variables easily observable at the time of choice. In the context of major choice, the information set involves some knowledge of individual talent for mathematics and non mathematics intensive activities, as well as some anticipation of future labor market conditions and the prices of talent.\footnote{We can allow the information set to contain a vector~$W$ of observed exogenous variables to increase the informativeness of our characterization, without changing the analysis: the characterization in Theorem~\ref{thm:charIF} would then be conditional on~$W$.}

\begin{assumption}[Selection]
\label{ass:if}
The utility function~$u$ on~$\mathcal Y\times\{0,1\}\times\mathcal Z$ is continuous and increasing in its first argument, and such that \[\mathbb E[u(Y_d,d,Z)|\mathcal I]>\mathbb E[u(Y_{1-d},1-d,Z)|\mathcal I]\Rightarrow D=d,\] where the conditional expectations are well defined, $\mathcal I$ is the sigma-algebra characterizing the agent's information set at the time of sector choice, $Z$ and~$D$ are $\mathcal I$-measurable.
\end{assumption}

Sector specific utility~$u(y,d,z)$ can be rationalized by different amenities. Women applicants may prefer Sector~$0$ because the proportion of women in the faculty is larger, hence female specific amenities are better provided. Sector specific utility may also be rationalized with reference dependence (\cite{Thaler:80} and \cite{TK:91}) based on gender profiling: social conditioning makes women prefer Sector~0.

The model characterized by Selection Assumption~\ref{ass:if} is an extended Roy model, in the sense that in each sector, and for each discrete socio-economic group, the utility is a deterministic function of the potential outcome and the non pecuniary vector of variables~$Z$. The utility may depend on unobservables (hence the dependence on Sector~$d$), but there is no unobservable heterogeneity at the individual level, beyond potential outcome~$Y_d$. 

%%%%%%%%%

% SMIV

%%%%%%%%%

\subsection{Stochastically Monotone Instrumental Variables}
\label{subsec:SMIV-IF}

Without further restrictions, the trivial choice~$u(y,d,z)=y$ (traditional Roy model) would rationalize any data generating process under Assumptions~\ref{ass:PO} and~\ref{ass:if}. Indeed, for any given random vector $(Y,D)$ generating observations, the choice~$Y_0=Y_1:=Y$ trivially satisfies Assumptions~\ref{ass:PO} and~\ref{ass:if}. To restore testability, we introduce a shape restriction on the joint distribution of potential expected utilities~$(\mathbb E[u(Y_0,0,Z)\vert \mathcal I],\mathbb E[u(Y_1,1,Z)\vert \mathcal I])$. 
A traditional approach to restoring testability without parametric restrictions is to allow some observed covariates to affect sector selection only. However, such restrictions are difficult to justify in the context of college major choice. Parental educational attainment is likely to be correlated with unobserved parental cognitive and non cognitive investments in their children (see \cite{Card:2001}). Distance to college and other instruments designed to affect educational attainment choices are not suitable for major choice. Other variables that are very relevant to a woman's major choice, such as the proportion of women on the faculty, are expected to affect potential outcomes for women as well as their choices. We resort instead to a weaker instrumental notion, where the instrument~$Z$ may affect the joint distribution of potential outcomes, but only in one direction, in terms of first order stochastic ordering. 

\begin{assumption}
\label{ass:MIV}
The vector 
$
(U_1,U_2):=(\mathbb E[u(Y_0,0,Z)\vert \mathcal I],\mathbb E[u(Y_1,1,Z)\vert \mathcal I])
$
is such that for each~$\nu\in\mathbb R$, $\mathbb P(U_1\leq \nu,U_2\leq \nu\vert Z=z)$ is monotonically non increasing in~$z$.
\end{assumption}

Assumption~\ref{ass:MIV} holds 
if the vector $(\mathbb E[u(Y_0,0,Z)\vert \mathcal I],\mathbb E[u(Y_1,1,Z)\vert \mathcal I])$ is stochastically monotone\footnote{See Appendix~\ref{app:SM} for a definition.} with respect to~$Z$. The case $\mathbb E[u(Y_d,d,Z)\vert \mathcal I]=\mathbb E[u(Y_d,d,Z)\vert Z]+V_d$, for~$d\in\{0,1\},$ with $(V_0,V_1)\perp Z$, coupled with monotonicity in~$z$ of $\mathbb E[u(Y_d,d,Z)\vert Z=z],$ for $d\in\{0,1\},$ is a special case of Assumption~\ref{ass:MIV}. Assumption~\ref{ass:MIV} also holds if the utility~$u$ is quasi-linear, increasing in~$z$, and the vector of expected potential outcomes $(\mathbb E[Y_0\vert \mathcal I],\mathbb E[Y_1\vert \mathcal I])$ is stochastically monotone with respect to~$Z$.

In the context of women's major choice, expected support and amenities for female students in STEM fields are likely to increase with the presence of female faculty and role models in STEM fields and in the educational level of the student's mother. Hence expected utility is likely to increase with the sector selection variables~$Z$. Combined with the assumption that more female faculty or a higher educational attainment of the mother cannot hurt a woman's earnings prospects, it yields Assumption~\ref{ass:MIV}.

%%%%%%%%%%%%%%

% CHARACTERIZATION

%%%%%%%%%%%%%%

\subsection{Characterization of the set of utility functions that rationalize the data}
\label{subsec:char}

We characterize the set of all utility functions that can rationalize the data under Assumptions~\ref{ass:PO}, \ref{ass:if}, and~\ref{ass:MIV}. This will be the content of Theorem~\ref{thm:charIF}. We start with a formal definition of the set of utility functions that rationalize the data under the model assumptions.\footnote{The identified set exhausts all the information contained in the model. As such, it is also known in the literature as {\em sharp identified region}.}

\begin{definition}[Identified Set]\label{def:idsetIF}
For any $\pi\in\Pi$, we call~$\mathcal U(\pi)$ the collection of functions $u:\mathcal Y\times\{0,1\}\times\mathcal Z\rightarrow\mathbb R$, continuous and increasing in their first argument, such that there exists a random vector $(Y_0,Y_1,D,Z)$ where $((1-D)Y_0+DY_1,D,Z)$ has distribution~$\pi$ and Assumptions~\ref{ass:PO}, \ref{ass:if}, and~\ref{ass:MIV} are satisfied for some~$\sigma$-algebra~$\mathcal I$ such that~$Z$ is~$\mathcal I$-measurable.
\end{definition}

We are interested in characterizing the identified set~$\mathcal U(\pi)$ with moment inequalities. Under Assumptions~\ref{ass:PO} and~\ref{ass:if}, when~$D=1$, $\mathbb E[u(Y,D,Z)\vert\mathcal I]=\mathbb E[u(Y_1,1,Z)\vert\mathcal I]\geq \mathbb E[u(Y_0,0,Z)\vert\mathcal I]$, and when~$D=0$, $\mathbb E[u(Y,D,Z)\vert\mathcal I]=\mathbb E[u(Y_0,0,Z)\vert\mathcal I]\geq \mathbb E[u(Y_1,1,Z)\vert\mathcal I]$. Hence Assumptions~\ref{ass:PO} and~\ref{ass:if} imply that
\begin{eqnarray}\label{eq:maxIF}
\mathbb E[u(Y,D,Z)\vert\mathcal I]=\max\{\mathbb E[u(Y_0,0,Z)\vert\mathcal I],\mathbb E[u(Y_1,1,Z)\vert\mathcal I]\}.
\end{eqnarray}
This relation is the key to deriving testable implications of restrictions on the joint distribution of potential outcomes. This relation also provides a direct proof of the identification of the cost function under the assumptions of \cite{HM:2011}, as described in Appendix~\ref{app:HM}. 
Equation~(\ref{eq:maxIF}) also implies that~$\mathbb E[u(Y,D,Z)\vert Z=z]$ is monotonically non decreasing in~$z$, which can be expressed as a collection of moment inequalities involving the infinite dimensional parameter~$u$. Conversely, monotonicity of~$\mathbb E[u(Y,D,Z)\vert Z=z]$ with respect to~$z$ is shown to imply that~$u$ can rationalize the data, i.e., that we can construct a vector of potential outcomes~$(Y_0,Y_1)$ such that Assumptions~\ref{ass:PO}, \ref{ass:if}, and~\ref{ass:MIV} are satisfied for some information set~$\mathcal I$. This discussion is formalized in the next theorem (proved in the appendix).

\begin{theorem}[Characterization under Imperfect Foresight]
\label{thm:charIF}
For any $\pi\in\Pi$, the identified set~$\mathcal U(\pi)$ is equal to the set of functions~$u$ on~$\mathcal Y\times\{0,1\}\times\mathcal Z$, continuous and increasing in their first argument, and such that for all ordered pairs $z\geq\tilde z$ on the support of $Z$, 
$\mathbb E(u(Y,D,Z)\vert Z=z)\geq\mathbb E(u(Y,D,Z)\vert Z=\tilde z)\geq\mathbb E(u(\underline b,1-D,Z)\vert Z=\tilde z)$, where $(Y,D,Z)$ is any vector with distribution~$\pi$. 
\end{theorem}

Theorem~\ref{thm:charIF} shows that the identified set of Definition~\ref{def:idsetIF} is characterized by monotonicity of conditional expected realized utilities with respect to~$Z$. A confidence region for the true utility can be obtained by collecting all the utility functions such that the hypothesis of monotonicity of~$\mathbb E[u(Y,D,Z)\vert Z=z]$ is not rejected. We provide two applications of Theorem~\ref{thm:charIF} to operationalize it in our empirical setting. First we propose a direct application to the case, where the utility is parametric (i.e., known up to a finite dimensional parameter vector). Then, we provide closed form expressions for the sharp bounds on the implied gender specific cost of STEM, under additional shape restrictions.

% PARAMETRIC UTILITY

\subsubsection*{Parametric utility function}

The most direct application of Theorem~\ref{thm:charIF} is to the case of parametric utility~$u(y,d,z;\theta)$ for a finite dimensional parameter~$\theta\in\Theta\subseteq\mathbb R^{l_\theta}$. In that case, the identified set
$\Theta_I(\pi):=\{\theta\in\Theta:\; u(\cdot,\cdot,\cdot;\theta)\in\mathcal U(\pi)\}$
is characterized by monotonicity of~$\mathbb E[u(Y,D,Z;\theta)\vert Z=z]$ with respect to~$z$. Hence, a confidence region for the true value~$\theta_0$ of the parameter
can be obtained by collecting all~$\theta\in\Theta$ such that the hypothesis of monotonicity of~$\mathbb E[u(Y,D,Z;\theta)\vert Z=z]$ with respect to~$z$ is not rejected.

% CLOSED FORM COST

\subsubsection*{Bounds on the cost function}
We now seek to characterize a minimal cost function that rationalizes the data under the imperfect foresight Roy model in closed form. We target our model specifically to explain under representation of women in STEM with a specific burden perceived by women in STEM education and STEM professional activities. To that end, we simplify our utility model to allow for an effect of our instrument in the STEM sector only. This is particularly relevant in case the instrument measures the influence of role models in STEM only. We therefore work on a restricted domain to obtain bounds in closed form.

\begin{definition}
\label{def:Cif} Call~$\tilde{\mathcal Z}$ the subset of~$\mathcal Z$ such that $u(y,1,z)\leq u(y,0,z)=y,$ all~$(y,z)\in\mathcal Y\times\tilde{\mathcal Z}$.
\end{definition}

In the rest of this section, we assume that the set~$\tilde{\mathcal  Z}$ of Definition~\ref{def:Cif} is non empty. We start from Theorem~\ref{thm:charIF}, which characterizes cost functions that rationalize the data with the monotonicity in~$z$ of the conditional expectation~$\mathbb E[u(Y,D,Z)\vert Z=z]$. This is equivalent to monotonicity in~$z$ of~$\mathbb E[Y-DC(Y,Z)\vert Z=z]$, with~$C(y,z):=y-u(y,1,z)$. Since $C$ is non negative on~$\mathcal Y\times\tilde{\mathcal Z}$, we seek functions~$C$ such that $\mathbb E[Y-DC(Y,Z)\vert Z=z]=\inf\{\mathbb E[Y\vert Z=\tilde z];\,\tilde z\geq z\},$ the lower monotone envelope of~$\mathbb E[Y\vert Z=z]$. 
As shown in the proof of Corollary~\ref{cor:cf}, this yields the moment equality constraint $\mathbb E[C(Y,Z)\vert D=1,Z=z]=\underline C(z)$, where~$\underline C$ is defined in~(\ref{eq:LB}) below. The upper bound~$\bar C(z)$ is obtained using the worst case bound.

\begin{corollary}\label{cor:cf}
Under Assumptions~\ref{ass:PO}, \ref{ass:if}, and~\ref{ass:MIV}, the cost function~$C(y,z)=y-u(y,1,z)$ satisfies~$\underline C(z)\leq\mathbb E[C(Y,Z)\vert D=1,Z=z]\leq \bar C(z)$, where~$\underline C$ and~$\bar C$ are defined for any~$z\in\tilde{\mathcal Z}$ by\footnote{We use the convention~$x/0=+\infty$ for any~$x\geq0$.} 
{\scriptsize
\begin{eqnarray}
\begin{array}{lll}
\label{eq:LB}
\underline C(z) & = & \mathbb P(D=1\vert Z=z)^{-1}\left(\mathbb E[Y\vert Z=z]-\inf_{\tilde z\geq z,\,\tilde z\in\tilde{\mathcal Z}}\mathbb E[Y\vert Z=\tilde z]\right)1\{\mathbb P(D=1\vert Z=z)>0\},\\
\bar C(z) & = & \mathbb P(D=1\vert Z=z)^{-1}\left(\mathbb E[Y\vert Z=z]-\sup_{\tilde z\leq z,\,\tilde z\in\tilde{\mathcal Z}}\mathbb E[Y(1-D)+\underline bD\vert Z=\tilde z]\right).
\end{array}
\end{eqnarray}
}
\end{corollary}

Corollary~\ref{cor:cf} yields closed form bounds on the cost function in case of quasi-linear utility~$u(y,1,z)=y-C(z)$. Corollary~\ref{cor:cf} also yields implications on the testability of the extended Roy model with imperfect foresight. If~$\pi$ is such that~$\mathbb P(D=1\vert Z=z)>0$ for all~$z$ on the support of~$Z$, then~$\underline C(z)$ is well defined, and the utility function defined by~$u(y,0,z):=y$ and~$u(y,1,z):=\underline C(z)$ is in~$\mathcal U(\pi)$. However, if~$\pi$ is such that $\mathbb P(D=1\vert Z=z)=\mathbb P(D=1\vert Z=\tilde z)=0$ and $\mathbb E[Y\vert Z=z]<\mathbb E[Y\vert Z=\tilde z]$ for $z\geq\tilde z$ on the support of~$Z$, then Assumptions~\ref{ass:PO}, \ref{ass:if} and~\ref{ass:MIV} are jointly rejected.

%%%%%%%%%%%%%%%%%%%%%%
%%%%%%%%%%%%%%%%%%%%%%

%    PERFECT FORESIGHT

%%%%%%%%%%%%%%%%%%%%%%
%%%%%%%%%%%%%%%%%%%%%%%

\section{Extended Roy model with perfect foresight}
\label{sec:frame}

The special case, where potential earnings are measurable with respect to the decision maker's information set has important implications and deserves special treatment. The empirical contents of the imperfect and perfect foresight models differ significantly. In the imperfect foresight case, sharp testable implications of the model specification take the form of conditional mean monotonicity constraints. In the perfect foresight case, they take the form of stochastic monotonicity constraints, and the resulting sharp bounds can therefore be considerably tighter. Another significant difference is the richer structural interpretation of the revealed costs of STEM as a compensating differential  in the perfect foresight case.

\subsection{Behavioral model specification}
\label{subsec:roy}

Although the perfect foresight case is nested within the imperfect foresight one described in Section~\ref{sec:if}, we state assumptions and results in full for greater clarity. The decision by the individual is based on the comparison between utilities~$u(Y_d,d,Z)$ for~$d\in\{0,1\}$, where $Y_d$ is the potential outcome (wage) in Sector~$d$ and $Z$ is the (possibly vector) value of the non pecuniary factor that may also affect utility. 

\begin{assumption+}{\ref{ass:if}$'$}[Selection]
\label{ass:sel}
The utility function~$(y,d,z)\mapsto u(y,d,z)$ on~$\mathcal Y\times\{0,1\}\times\mathcal Z$ is continuous and increasing in its first argument, and satisfies 
\[
u(Y_d,d,Z)>u(Y_{1-d},1-d,Z)\Rightarrow D=d, \mbox{ for }d\in\{0,1\}.
\]
\end{assumption+}

In case of quasi-linear utility~$u(y,d,z);=y-C(d,z)$, Assumption~\ref{ass:sel} is equivalent to~$Y_d-C(d,Z)>Y_{1-d}-C(1-d,Z)\Rightarrow D=d$, for $d\in\{0,1\},$ which is equivalent to the traditional extended Roy model selection, under perfect foresight, except for the fact that we don't impose any tie-breaking rule. However, in many applications, the non pecuniary cost may depend on potential outcomes, so that considering non quasi-linear utility comparisons is important. For example, we expect the costs incurred by women studying or working in the mathematics intensive sector to be less pronounced for women with higher mathematics ability, hence decreasing in~$Y_1$. In the more general case (beyond quasi-linear utility), Assumption~\ref{ass:sel} is equivalent to
\begin{eqnarray}
\label{eq:sel-C}
Y_1-C(Y_1,Z)>Y_0\Rightarrow D=1, \mbox{ and } Y_1-C(Y_1,Z)<Y_0\Rightarrow D=0,
\end{eqnarray}
with~$C(y,z)=y-u^{-1}(u(y,1,z),0,z)$, where~$u^{-1}$ denotes the inverse with respect to the first argument. Assumption~\ref{ass:sel}, therefore, describes a perfect foresight extended Roy model, where the non pecuniary cost may depend on the potential outcome.\footnote{More generally, the cost function~$C(y,z)=y-u^{-1}(u(y,1,z),0,z)$ depends on potential outcomes unless the constraint~$\partial u(y,1,z)/\partial y=\partial u(u^{-1}(u(y,1,z),0,z),0,z)/\partial y$ holds.}

In the perfect foresight case, the cost function~$C$ can be clearly interpreted as a compensating wage differential. Suppose women perceive inferior amenities in the STEM sector. Call~$(y,z)\mapsto\tilde C(y,z)$ the compensating differential defined by $u(y,1,z)=u(y-\tilde C(y,z),0,z)$. Then $\tilde C(y,z)=y-u^{-1}(u(y,1,z),0,z)=C(y,z)$. Hence, $C(y,z)$ in Assumption~\ref{ass:sel} is a monetary adjustment that makes women, whose talents entitle them to identical (uncompensated) wages in both sectors, indifferent between the two sectors. As defined, it is a willingness to pay for the better amenities of the non-STEM sector, or equivalently, the equivalent variation to a move from non-STEM to STEM.

%%%%%%%%%

% SMIV

%%%%%%%%%

\subsection{Stochastically Monotone Instrumental Variables}
\label{subsec:SMIV}

As in the general case, without further restrictions, the trivial choice~$u(y,d,z)=y$ (traditional Roy model) would rationalize any data generating process under Assumptions~\ref{ass:PO} and~\ref{ass:sel}. Indeed, for any given random vector $(Y,D)$ generating observations, the choice~$Y_0=Y_1:=Y$ trivially satisfies Assumptions~\ref{ass:PO} and~\ref{ass:sel}. The shape restriction inspired by the influence of role models takes the following form in case of perfect foresight.

\begin{assumption+}{\ref{ass:MIV}$'$}
\label{ass:mMax}
The random vector $(u(Y_0,0,Z),u(Y_1,1,Z))$ is such that for each~$\nu\in\mathbb R$, the quantity~$\mathbb P(u(Y_0,0,Z)\leq \nu,u(Y_1,1,Z)\leq \nu\vert Z=z)$ is monotonically non increasing in~$z$.
\end{assumption+}

Assumption~\ref{ass:mMax} allows dependence of potential outcomes on the instrument. Assumption~\ref{ass:mMax} is weaker than stochastic monotonicity of the vector of potential utilities in~$z$, as defined in Appendix~\ref{app:SM}. Indeed, stochastic monotonicity implies (but is not equivalent to) monotonicity in~$z$ of the quantity~$\mathbb P(u(Y_0,0,Z)\leq \nu_0,u(Y_1,1,Z)\leq \nu_1\vert Z=z)$ for all~$(\nu_0,\nu_1)\in\mathbb R^2$. Assumption~\ref{ass:mMax} only involves constraints for a scalar~$\nu$ and is hence weaker. Assumption~\ref{ass:mMax} is in the same spirit, but is not directly comparable to the monotone instrumental variable (MIV) restriction of \cite{MP:2000}. Unlike MIV, Assumption~\ref{ass:mMax} places restrictions on the joint distribution of potential outcomes, as opposed to the marginals only, and drives our characterization of the model's empirical content in Theorem~\ref{thm:char}. 

In the context of women's major choice, we expect support and amenities for female students in STEM fields to increase with the presence of female faculty and role models in STEM fields and in the educational level of the student's mother. Hence we expect utility to increase with the sector selection variables~$Z$. Combined with the assumption that more female faculty or a higher educational attainment of the mother cannot hurt a woman's earnings prospects, it yields Assumption~\ref{ass:mMax}.

%%%%%%%%%%%%%%%%%%%%

% CHARACTERIZATION

%%%%%%%%%%%%%%%%%%%%

\subsection{Characterization of the set of utility functions that rationalize the data}
%\label{subsec:char-IF}

We characterize the set of all utility functions that can rationalize the data under Assumptions~\ref{ass:PO}, \ref{ass:sel}, and~\ref{ass:mMax}. This will be the content of Theorem~\ref{thm:char}. We start with a formal definition of the set of utility functions that rationalize the data under the model assumptions.

\begin{definition+}{\ref{def:idsetIF}$'$}[Identified Set under perfect foresight]\label{def:idset}
For any $\pi\in\Pi$, we call~$\mathcal U^\prime(\pi)$ the collection of functions $u:\mathcal Y\times\{0,1\}\times\mathcal Z\rightarrow\mathbb R$, such that there exists a random vector $(Y_0,Y_1,D,Z)$ where $((1-D)Y_0+DY_1,D,Z)$ has distribution~$\pi$ and Assumptions~\ref{ass:PO}, \ref{ass:sel}, and~\ref{ass:mMax} are satisfied.
\end{definition+}

We are interested in characterizing the identified set~$\mathcal U^\prime(\pi)$ with moment inequalities. 
Under Assumptions~\ref{ass:PO} and~\ref{ass:sel}, when~$D=1$, $u(Y,D,Z)=u(Y_1,1,Z)\geq u(Y_0,0,Z)$, and when~$D=0$, $u(Y,D,Z)=u(Y_0,0,Z)\geq u(Y_1,1,Z)$. Hence we have 
\begin{eqnarray}
\label{eq:max}
u(Y,D,Z)=\max\{u(Y_1,1,Z),u(Y_0,0,Z)\}.
\end{eqnarray}
Hence, Assumptions~\ref{ass:PO}, \ref{ass:sel}, and~\ref{ass:mMax} imply that~$\mathbb P(u(Y,D,Z)\leq \nu\vert Z=z)$ is monotonically non increasing in~$z$, for all~$\nu\in\mathbb R$, which is equivalent (by definition) to stochastic monotonicity of~$u(Y,D,Z)$ with respect to~$Z$, and which can be expressed as a collection of moment inequalities involving the infinite dimensional parameter~$u$. Conversely, stochastic monotonicity of~$u(Y,D,Z)$ with respect to~$Z$ is shown to imply that~$u$ can rationalize the data, i.e., that we can construct a vector of potential outcomes~$(Y_0,Y_1)$ such that Assumptions~\ref{ass:PO}, \ref{ass:sel}, and~\ref{ass:mMax} are satisfied. This discussion is formalized in the next theorem (proved in the appendix).

\begin{theorem}[Characterization]\label{thm:char}
For any $\pi\in\Pi$, the identified set~$\mathcal U^\prime(\pi)$ is equal to the collection of functions $u:\mathcal Y\times\{0,1\}\times\mathcal Z\rightarrow\mathbb R$, continuous and increasing in the first argument, such that for any $(Y,D,Z)$ with distribution~$\pi$, all~$\nu\in\mathbb R$ and all pairs $z\geq\tilde z$ in~$\mathcal Z$, 
$u(Y,D,Z)\geq u(\underline b,1-D,Z)$ a.s., and
\begin{eqnarray}\label{eq:char}
\mathbb P(u(Y,D,Z)>\nu\vert Z=z) \, \geq \, \mathbb P(u(Y,D,Z)>\nu\vert Z=\tilde z).
\end{eqnarray}
\end{theorem}

Theorem~\ref{thm:char} shows that the identified set of Definition~\ref{def:idset} is characterized by stochastic monotonicity of~$u(Y,D,Z)$ with respect to~$Z$, which is equivalent to the collection of moment inequalities in~(\ref{eq:char}). A confidence region for the true utility can be obtained by collecting all the utility functions such that the hypothesis of stochastic monotonicity of~$u(Y,D,Z)$ is not rejected. We provide two corollaries to Theorem~\ref{thm:char} to operationalize it in our empirical setting. First we propose a direct application to the case, where the utility is parametric (i.e., known up to a finite dimensional parameter vector). Then, we provide closed form expressions for the sharp bounds on the implied gender specific cost of STEM, under additional shape restrictions.

% PARAMETRIC UTILITY

\subsubsection*{Parametric utility function}

The most direct application of Theorem~\ref{thm:char} is to the case of parametric utility~$u(y,d,z;\theta)$ for a finite dimensional parameter~$\theta\in\Theta\subseteq\mathbb R^{l_\theta}$. In that case, the identified set
$\Theta_I^\prime(\pi):=\{\theta\in\Theta:\; u(\cdot,\cdot,\cdot;\theta)\in\mathcal U^\prime(\pi)\}$
is characterized by the collection of moment inequalities in~(\ref{eq:char}). Moreover, a confidence region for the true value~$\theta_0$ of the parameter
can be obtained by collecting all~$\theta\in\Theta$ such that the hypothesis of stochastic monotonicity of~$u(Y,D,Z;\theta)$ with respect to~$Z$ is not rejected.

% CLOSED FORM COST

\subsubsection*{Closed form expression for the bounds on the cost function}

Assumption~\ref{ass:sel} is equivalent to~(\ref{eq:sel-C}) with cost function~$C(y,z)=y-u^{-1}(u(y,1,z),0,z)$. We derive closed form bounds on the cost function~$C$ under the following shape restrictions on the utility functions.

\begin{definition+}{\ref{def:Cif}$'$}
\label{def:C}
Call~$\tilde{\mathcal Z}^\prime$ the subset of~$\mathcal Z$ such that, for each~$(y,z)\in\mathcal Y\times\tilde{\mathcal Z}^\prime$, (1)~$u(y,0,z)\geq u(y,1,z),$ and (2)~$u(y,0,z)$ is non increasing in~$z$.
\end{definition+}

On~$\tilde{\mathcal Z}^\prime$, which we assume non empty for the remainder of this section, unobserved amenities are superior in Sector~$0$ for all individuals in the socio-economic group of interest. This is empirically relevant in our application, where~$z=(z_0,z_1)$ is the vector of proportions of women on the the faculty in each sector, and~$u(y,d,z)=u(y,z_d)$. Then, Condition~(1) becomes~$u(y,z_1)\leq u(y,z_0)$, which is plausible, since STEM departments have far lower proportions of women on the faculty than non STEM ones. As for Condition~(2), it is satisfied if the utility function has a satiation point (in~$z_d$) below the realized proportions of women on the faculty of non STEM departments.

On~$\tilde{\mathcal Z}^\prime$, we have:
\begin{eqnarray*}
Y & \geq & Y-DC(Y,Z) \\
& = & Y(1-D)+Du^{-1}(u(Y,1,Z),0,z) \\
& = & \left\{ \begin{array}{lll}
u^{-1}(u(Y_1,1,Z),0,z) & \mbox{if} & D=1,\\ Y_0 & \mbox{if} & D=0
\end{array} \right. \\
& \geq & Y_0,
\end{eqnarray*}
where the first inequality and the first equality follow from the definitions of~$C$ and~$\tilde{\mathcal Z}^\prime$, the second equality follows from Assumption~\ref{ass:PO}, and the last inequality follows from Assumption~\ref{ass:sel} and the fact that~$u(y,0,z)$ is increasing in~$y$.
Hence, for each~$y\in\mathcal Y$, we have
\begin{eqnarray*}
\mathbb P(Y\leq y\vert z) \; \leq \; \mathbb P(Y-DC(Y,Z)\leq y\vert z) \; \leq \; \mathbb P(Y\leq y,D=0\vert z)+\mathbb P(D=1\vert z)1\{y\geq\underline b\},
\end{eqnarray*}
where the last inequality uses standard worst-case bounds for the distribution of potential outcome~$Y_0$. 

Now, the middle term is equal to
\begin{eqnarray*}
\mathbb P(Y-DC(Y,Z)\leq y\vert Z=z) & = & \mathbb P(\max\{Y_0,u^{-1}(u(Y_1,1,z),0,z)\}\leq y\vert Z=z) \\
& = & \mathbb P(\max\{u(Y_0,0,z),u(Y_1,1,z)\}\leq u(y,0,z)\vert Z=z) \\
& = & \mathbb P(u(Y_0,0,z)\leq u(y,0,z) \mbox{ and } u(Y_1,1,z)\leq u(y,0,z)\vert Z=z),
\end{eqnarray*}
which is monotone non increasing in~$z\in\tilde{\mathcal Z}^\prime$ for each~$y$ by Assumptions~\ref{ass:mMax}, and is right-continuous in~$y$. 
Hence, the random variable~$Y-DC(Y,Z)$ is stochastically monotone non decreasing with respect to~$Z$. Heuristically, if realized outcomes~$Y$ was stochastically monotone non decreasing with respect to~$Z$, the function~$C=0$ would rationalize the data, and would then be the desired lower bound. In general, the lower bound ~$\underline C$ for the cost function is such that the distribution of~$Y-D\underline C(Y,Z)$ is the monotone lower envelope of the distribution of~$Y$, as defined by the first display in~(\ref{eq:bds}) below. The second display in~(\ref{eq:bds}) defines the upper envelope of the worst case bound in the same way.
\begin{eqnarray}
\label{eq:bds}
\begin{array}{lll}
\underline F(y\vert z) & := & \lim_{\tilde y\downarrow y}\sup\left\{\mathbb P(Y\leq \tilde y\vert \tilde z):\tilde z\in\tilde{\mathcal Z}^\prime,\tilde z\geq z\right\}, \\ \\
\bar F(y\vert z) & := & \lim_{\tilde y\downarrow y} \inf\left\{\mathbb P(Y\leq \tilde y,D=0\vert \tilde z)+\mathbb P(D=1\vert \tilde z)1\{\tilde y\geq\underline b\} : \tilde z\in\tilde{\mathcal Z}^\prime,\tilde z\leq z \right\}.
\end{array}
\end{eqnarray}
Given the stochastic monotonicity of~$Y-DC(Y,Z)$ with respect to~$Z$, we therefore have
\begin{eqnarray}
\label{eq:poitwise}
\underline F(y\vert z) \; \leq \; \mathbb P(Y-DC(Y,Z)\leq y\vert z) \; \leq \; \bar F(y\vert z),
\end{eqnarray}
for all~$(y,z)\in\mathcal Y\times\tilde{\mathcal Z}^\prime,$ whenever Assumptions~\ref{ass:PO}, \ref{ass:sel}, and~\ref{ass:mMax} hold.

Lemma~\ref{lemma:env} below shows that the envelopes defined in display~(\ref{eq:bds}) are themselves cumulative distribution functions, and that they are the bounds of the set of cumulative distribution functions with the desired monotonicity properties.
\begin{lemma}[Monotone Envelope]\label{lemma:env}
The following hold for~$\underline F$ and~$\bar F$ defined in~(\ref{eq:bds}): 
\begin{enumerate}
\item For each~$z\in\tilde{\mathcal Z}^\prime$, $y\mapsto \underline F(y\vert z)$ and~$y\mapsto \bar F(y\vert z)$ are cumulative distribution functions. 
\item If $y\mapsto \tilde F(y\vert z)$ is a cumulative distribution function that satisfies 
\[
\mathbb P(Y\leq y\vert z)\leq\tilde F(y\vert z)\leq \mathbb P(Y\leq y,D=0\vert z)+\mathbb P(D=1\vert z)1\{y\geq\underline b\}\] 
for all~$(y,z)\in\mathcal Y\times\tilde{\mathcal Z}^\prime,$ and $z\mapsto\tilde F(y\vert z)$ is non increasing in~$z$ for all~$y,$ then for all~$(y,z)\in\mathcal Y\times\tilde{\mathcal Z}^\prime,$ 
\[
\underline F(y\vert z)\leq \tilde F(y\vert z)\leq \bar F(y\vert z).
\]
\end{enumerate}
\end{lemma}

Equation~(\ref{eq:poitwise}) yields the testable implication for our model that bounds~$\underline F(y\vert z)$ and~$\bar F(y\vert z)$ cannot cross.
Since~$\mathbb P(Y-DC(Y,Z)\leq y\vert z) = \mathbb P(Y-C(Y,Z)\leq y, D=1\vert z) + \mathbb P(Y\leq y, D=0\vert z),$ and since~$\mathbb P(Y-C(Y,Z)\leq y, D=1\vert z)$ is non decreasing in~$y$, (\ref{eq:poitwise}) also yields the bounds
\begin{eqnarray}
\label{eq:pointwise}
\begin{array}{lll}
L(y\vert z) & := & \sup_{\tilde y\leq y}\left\{\underline F(\tilde y\vert z)-\mathbb P(Y\leq \tilde y,D=0\vert z) \right\} \\ \\
& \leq & \mathbb P(Y-C(Y,Z)\leq y, D=1\vert z)  \\ \\
& \leq & \inf_{y\leq \tilde y}\left\{\bar F(\tilde y\vert z)-\mathbb P(Y\leq \tilde y,D=0\vert z) \right\} \; := \; U(y\vert z).
\end{array}
\end{eqnarray}
Finally, we have the following corollary of Theorem~\ref{thm:char}, proved in the Appendix, which establishes bounds on the cost functions that correspond to utilities in the identified set of Definition~\ref{def:idset}.

\begin{corollary}[Bounds on the Cost Function]
\label{cor:pointwise}
Under Assumptions~\ref{ass:PO}, \ref{ass:sel}, and~\ref{ass:mMax}, the cost function~$C(y,z)=y-u^{-1}(u(y,1,z),0,z)$ satisfies for all~$(y,z)\in\mathcal Y\times\tilde{\mathcal Z}^\prime$:
\begin{eqnarray}
\label{eq:CFB}
\underline C(y,z) \; := \; y-L_-\left(F_1(y\vert z)\vert z\right)
\; \leq \; C(y,z)
\; \leq \; y-U^-\left(F_1(y\vert z)\vert z\right) \; := \; \bar C(y,z),
\end{eqnarray}
where~$F_1(y\vert z):=\mathbb P(Y\leq y,D=1\vert Z=z)$ and~$U^-$ and~$L_-$, defined respectively by~$U^-(x|z):=\inf\{y\in\mathbb R:U(y|z)\geq x\}$ and~$L_-(x|z):=\sup\{y\in\mathbb R:L(y|z)\leq x\}$, are generalized inverses  of~$U$ and~$L$ from~(\ref{eq:pointwise}). 
\end{corollary}

% SHARPNESS OF CLOSED FORM BOUNDS

\subsubsection*{Sharpness of the closed form bounds}

The bounds of Corollary~\ref{cor:pointwise} may not be attained under the additional shape restrictions on utilities. We now turn to a characterization of assumptions on the cost function~$C(y,z)$ such that the bounds of Corollary~\ref{cor:pointwise} are attained (hence sharp). Since the result below relies on a testable, but high level regularity condition on the data generating process~$\pi$, we later complement it with an iterative procedure to tighten the bounds in case the regularity condition fails to hold.

Consider the following restricted set of data generating processes.

\begin{definition}
\label{def:adj}
Let~$\Pi^r$ be the set of distributions on~$\mathcal Y\times\{0,1\}\times\tilde{\mathcal Z}^\prime$ such that the functions~$F_1(\cdot\vert z)$ from Corollary~\ref{cor:pointwise}, $\underline F(\cdot\vert z)-\mathbb P(Y\leq \cdot, D=0\vert z)$ and~$\bar F(\cdot\vert z)-\mathbb P(Y\leq \cdot, D=0\vert z)$,
where~$\underline F$ and~$\bar F$ are defined in~(\ref{eq:bds}), are continuous and increasing for all~$z\in\tilde{\mathcal Z}^\prime$ and do not cross.
\end{definition}

The bounds of Corollary~\ref{cor:pointwise} are sharp under the regularity assumption of Definition~\ref{def:adj} and under the shape restrictions on the cost function used to derive~(\ref{eq:CFB}) and formalized in the following definition. 
\begin{definition}
\label{def:adj2}
For each~$\pi\in\Pi$, where~$\Pi$ is the set of distributions on~$\mathcal Y\times\{0,1\}\times\tilde{\mathcal Z}^\prime$, let~$\mathcal C(\pi)$ be the set of non negative functions on~$\mathcal Y\times\tilde{\mathcal Z}^\prime$ such that~$y\mapsto y-C(y,z),$ is increasing for all~$z$, $\mathbb P(Y\leq y,D=1\vert Z=z)$ is continuous and increasing in~$y$ for all~$z\in\tilde{\mathcal Z}^\prime$, and $\mathbb P(Y-DC(Y,Z)>y\vert Z=z)$ is monotone non decreasing in~$z$ for all~$y$, for any random vector~$(Y,D,Z)$ distributed according to~$\pi$.
\end{definition}

\begin{proposition}[Sharpness of the closed-form bounds]\label{prop:min}
Assume~$\pi\in\Pi^r$, where~$\Pi^r$ is defined according to Definition~\ref{def:adj}. Let~$\underline C$ and~$\bar C$ be defined in~(\ref{eq:CFB}) and let~$\mathcal C(\pi)$, defined according to Definition~\ref{def:adj2}, be non empty. Then~$\underline C\in\mathcal C(\pi)$, $\bar C\in\mathcal C(\pi)$, and all~$C\in\mathcal C(\pi)$ satisfy~$\underline C\leq C\leq \bar C$.
\end{proposition}

Sharpness of the bounds is understood here in the following way. Assumptions~\ref{ass:PO}, \ref{ass:sel}, \ref{ass:mMax}, and the conditions of Definition~\ref{def:C} imply~$C(y,z):=y-u^{-1}(u(y,1,z),0,z)\in\mathcal C(\pi)$. Moreover, Proposition~\ref{prop:min} shows that the bounds of Corollary~\ref{cor:pointwise} belong to~$\mathcal C(\pi)$. Finally, the utility functions~$\underline u:=y-d\underline C$ and~$\bar u:=y-d\bar C$ satisfy Assumptions~\ref{ass:PO}, \ref{ass:sel}, \ref{ass:mMax}, and the conditions of Definition~\ref{def:C}.

In case the data generating process~$\pi$ fails to satisfy the regularity condition of Definition~\ref{def:adj}, we propose an iterative procedure to sharpen the bounds. Let~$\underline C^{(0)}:= \underline C$, and~$\bar C^{(0)}:=\bar C$, 
as defined in Equation~(\ref{eq:CFB}) of Corollary~\ref{cor:pointwise}. Then, for each $y\in  \mathcal{Y}$,  $z \in \tilde{\mathcal Z}^\prime$, and $n=1,2, \ldots$, define the sequences:
\begin{eqnarray}\label{eq:def_Fn-}
\begin{array}{lll}
\underline F^{(n)}(y\vert z) & := & \lim_{\tilde y\downarrow y}\sup\left\{\mathbb P(Y- D \underline{C}^{(n-1)}(Y,Z) \leq \tilde y\vert \tilde z):\tilde z\in\tilde{\mathcal Z}^\prime,\tilde z\geq z\right\},\\ \\
L^{(n)}(y\vert z) & := & \sup_{\tilde y\leq y}\left\{\underline F^{(n)}(\tilde y\vert z)-\mathbb P(Y\leq \tilde y\vert z) \right\}, \\ \\
\underline C^{(n)}(y,z) \; &:=& \; y-L^{(n)}_{-}\left(F_1(y\vert z)\vert z\right),
\end{array}
\end{eqnarray}
where~$L_-$ is the generalized inverse of~$L$ as defined in Corollary~\ref{cor:pointwise}. 
Symmetrically, for the upper bound, define the sequences
\begin{eqnarray}\label{eq:def_Fn+}
\begin{array}{lll}
\bar F^{(n)}(y\vert z) & := & \lim_{\tilde y\downarrow y}\inf\left\{\mathbb P(Y- D \bar C^{(n-1)}(Y,Z) \leq \tilde y\vert \tilde z):\tilde z\in\tilde{\mathcal Z}^\prime,\tilde z\leq z\right\},\\ \\
U_{(n)}(y\vert z) & := & \inf_{\tilde y\leq y}\left\{\bar F^{(n)}(\tilde y\vert z)-\mathbb P(Y\leq \tilde y\vert z) \right\}, \\ \\
\bar C^{(n)}(y,z) \; &:=& \; y-U_{(n)}^{-}\left(F_1(y\vert z)\vert z\right),
\end{array}
\end{eqnarray}
where~$U^-$ is the generalized inverse of~$U$ as defined in Corollary~\ref{cor:pointwise}.

\begin{proposition}[Sharpness of the iterated bounds]\label{prop:iter}
Let~$\pi\in\Pi$, where~$\Pi$ is the set of distributions on~$\mathcal Y\times\{0,1\}\times\tilde{\mathcal Z}^\prime$. Let~$\mathcal C(\pi)$, defined according to Definition~\ref{def:adj2}, be non empty. The sequences~$(\underline C^{(n)})_n$ and~$(\bar C^{(n)})_n$ 
defined in~(\ref{eq:def_Fn-}) and~(\ref{eq:def_Fn+}) 
admit limits~$\underline C^{(\infty)}$ and~$\bar C^{(\infty)}$ respectively
$\in\mathcal C(\pi)$, $\bar C^{(\infty)}\in\mathcal C(\pi)$, 
and all~$C\in\mathcal C(\pi)$ satisfy~$\underline C^{(\infty)}\leq C\leq \bar C^{(\infty)}$.
\end{proposition}

Inference based on the sharp iterated bounds of Proposition~\ref{prop:iter} poses challenges that are beyond the scope of this paper. In Section~\ref{sec:sim}, we therefore base inference on the bounds of Proposition~\ref{prop:min}, which are only sharp under the additional regularity conditions of Definition~\ref{def:adj2}, but remain valid without them.

%%%%%%%%%%%%%%%%%%%%%%%%%%%%%%%%%%%%%%%%%%%%

%                                 INFERENCE

%%%%%%%%%%%%%%%%%%%%%%%%%%%%%%%%%%%%%%%%%%%%

\section{Inference}
\label{sec:inference}

This section proposes guidelines for inference procedures based on the identification results of Sections~\ref{sec:if} and~\ref{sec:frame}. All proposed inference is based on an i.i.d. sample of outcomes, decisions, covariates and instruments. In Section~\ref{sec:para}, we discuss parametric inference based on Theorems~\ref{thm:charIF} and~\ref{thm:char}. Confidence regions for utility parameters are obtained from the inversion of tests of regression and stochastic monotonicity in the imperfect and perfect foresight cases respectively. In Section~\ref{sec:npara}, we propose pointwise confidence regions for the implied non pecuniary cost of STEM based on the closed form bounds of Corollaries~\ref{cor:cf} and~\ref{cor:pointwise} for the imperfect and perfect foresight cases respectively. We propose using the intersection bounds methodology of \cite{CLR:2009}.

\subsection{Parametric inference}
\label{sec:para}

In Theorem~\ref{thm:charIF}, we characterize the set of utility functions that rationalize the data under the extended Roy model of Assumptions~\ref{ass:PO}, \ref{ass:if}, and~\ref{ass:MIV}. The characterization is regression monotonicity of realized utilities with respect to the instrument. Parametric inference on utilities can therefore be conducted by inverting a test of regression monotonicity. More precisely, assume utility takes the form~$u(Y,D,Z;\theta)$ for some parameter vector~$\theta\in\Theta$. 
In Section~\ref{sec:sim}, we will consider two standard parameterizations of utility, namely a quasi-linear (QL) and a constant elasticity of substitution (CES) form. In both cases, the non pecuniary component in utility is a disamenity from low representation of women among the faculty in STEM fields. The disamenity from belonging to the minority in STEM is assumed to vanish above a threshold~$\gamma$ (i.e.,~$z\geq\gamma$). For~$\alpha\in\mathbb R$,~$\beta>1$, and~$\gamma\in(0,0.5]$, define
\begin{eqnarray}
\label{eq:QL}
u_{\sc QL}(y,d,z;(\alpha,\beta,\gamma)) & := & y-d\,\alpha \left( 1-\frac{z}{\gamma} \right)^\beta1\{z\leq\gamma\}, \\
\label{eq:CES}
u_{\sc CES}(y,d,z;(\alpha,\beta,\gamma)) & := & \left( y^\frac{\beta-1}{\beta}-d\,\alpha \left( 1-\frac{z}{\gamma}\right)^\frac{\beta-1}{\beta}1\{z\leq\gamma\}\right)_+^\frac{\beta}{\beta-1}. 
\end{eqnarray}

Confidence regions for the true value of the parameter can be obtained by collecting all values of~$\theta=(\alpha,\beta,\gamma)$ such that the hypothesis of regression monotonicity
\begin{eqnarray*}
\mbox{H}_0: \mathbb E[u(Y,D,Z;\theta)\vert Z=z] \mbox{ is non decreasing in }z,
\end{eqnarray*}
is not rejected.
Early tests of regression monotonicity are proposed in \cite{GSvV:2000} and \cite{HH:2000}. \cite{Chetverikov:2013} and \cite{HLS:2016} offer power improvements.

In Theorem~\ref{thm:char}, we characterize the set of utility functions that rationalize the data under the extended Roy model of Assumptions~\ref{ass:PO}, \ref{ass:sel}, and~\ref{ass:mMax}. The characterization is stochastic monotonicity of realized utilities with respect to the instrument. Parametric inference on utilities can therefore be conducted by inverting a test of stochastic monotonicity. More precisely, assume utility takes the form~$u(Y,D,Z;\theta)$ for some parameter vector~$\theta\in\Theta$. For instance,~$u(y,z;\theta)$ takes quasi-linear or CES forms in~(\ref{eq:QL}) or~(\ref{eq:CES}) respectively. Confidence regions for the true value of the parameter can be obtained by collecting all values of~$\theta$ such that a hypothesis of stochastic monotonicity
\begin{eqnarray}
\mbox{H}^\prime_0: \mathbb P[u(Y,D,Z;\theta)\leq u\vert Z=z] \mbox{ is non increasing in }z\mbox{ for all } u,
\label{eq:SM}
\end{eqnarray}
is not rejected.
Early tests of stochastic monotonicity, namely \cite{LLW:2009}, \cite{DE:2012} produce limiting rejection rates equal to nominal size under the least favorable DGP. Hence, they can be conservative. More recent stochastic monotonicity tests provide power improvements via pre-estimation of contact sets, as in \cite{HLS:2016} and \cite{LSW:2018}, using directional differentiability of the least concave majorant operator, as in \cite{Seo:2018}, or adaptivity to smoothness in the conditional cdf, as in \cite{CWK:2021}. 

In what follows, we test both regression monotonicity and stochastic monotonicity using \cite{HLS:2016}.\footnote{\scriptsize We thank Yu-Chin Hsu, Chu-An Liu and Xiaoxia Shi for sharing their code.} The procedure is valid under continuity of the conditional mean~$z\mapsto\mathbb E[u(Y,D,Z;\theta)\vert Z=z]$ and conditional cdf~$z\mapsto\mathbb P(u(Y,D,Z;\theta)\leq u\vert Z=z)$ respectively. The sensitivity of inference results to the generalized moment selection procedure is usually the major concern with this type of procedure, see for instance \cite{CS:2018}\footnote{\scriptsize The generalized moment selection procedure, originally introduced in \cite{Hansen:2005}, \cite{GH:2009} and \cite{AS:2010}, increases the power of moment inequality tests, while controlling size, by pre-selecting inequalities that are close to binding. In the specific implementation of moment inequality testing in \cite{HLS:2016}, the threshold according to which moment inequalities are pre-selected depends on the user-chosen quantities~$\kappa_n$ and~$B_n$.}. We choose the recommended values for the user-chosen parameters governing the generalized moment selection in \cite{HLS:2016}, namely $B_n=0.85\ln n/\ln\ln n$ and $\kappa_n=0.15\ln n$. 

\subsection{Nonparametric bounds}
\label{sec:npara}

This subsection discusses inference based on the closed form bounds of Corollaries~\ref{cor:cf} and~\ref{cor:pointwise} in the imperfect and perfect foresight cases, respectively. The bounds are intersection bounds, and we apply the methodology of Section~4.3 of \cite{CLR:2009}. Implementation and code are taken from \cite{CKLR:2013}.

\subsubsection{Nonparametric bounds under imperfect foresight}

The bounds~(\ref{eq:LB}) of Corollary~\ref{cor:cf} are intersection bounds, and we propose applying the inference procedure proposed in \cite{CLR:2009}. We discuss implementation and applicability for the lower bound~$\underline C$ in~(\ref{eq:LB}). The upper bound is similar. The lower bound~$\underline C$ in~(\ref{eq:LB}) can be written
\begin{eqnarray*}
\underline C(z) & = & \sup_{\tilde z\geq z} \theta(z,\tilde z), 
\end{eqnarray*}
where
\begin{eqnarray*}
\theta(z,\tilde z) & = & \frac{1}{\mathbb P(D=1\vert Z=z)} \left( \mathbb E[Y\vert Z=z] - \mathbb E[Y\vert Z=\tilde z]  \right),
\end{eqnarray*}
using the notation of \cite{CLR:2009} (not to be confused with parameter~$\theta$ in the previous section). The function~$\theta(z,\tilde z)$ is estimated with a twice continuously differentiable kernel and a bandwidth sequence that guarantees undersmoothing, hence asymptotically negligible bias (we use the bandwidth recommendation on page~8 of \cite{CKLR:2013}). Under the latter, 
Bahadur representation~(2) page 1531 of \cite{KLX:2010} implies condition NK page~703 of \cite{CLR:2009} (as shown in Theorem~8 of Appendix~G in the online Supplemental Material). In turn, Condition NK implies validity of the bounds by Theorem~6 page~706 of \cite{CLR:2009}.

\subsubsection{Nonparametric bounds under perfect foresight}

We now propose an inference procedure for the minimal cost function of Corollary~\ref{cor:pointwise}. The upper bound can be treated symmetrically. For any given value of~$z\in\mathcal Z$, we seek a data driven function~$y\mapsto C_n(y,z)$ such that for each~$y\in\mathcal Y$,
\begin{eqnarray}
\label{eq:CR}
\lim_{n\rightarrow\infty}\mathbb P(C_n(y,z)\leq \underline C(y,z))\geq 1-\alpha,
\end{eqnarray}
for some pre-determined level of significance~$\alpha$. Define~$G(y\vert z,\tilde z):=\mathbb P(Y\leq y\vert \tilde z)-\mathbb P(Y\leq y,D=0\vert z)$. Call~$\hat G$ a non parametric estimator for~$G$, and define
\[
\hat G_-(x\vert z,\tilde z):=\sup\left\{ y\in\mathcal Y \; : \; \hat G(y\vert z,\tilde z)\leq x\right\}.
\] 
Finally, let~$\hat F_1$ be a nonparametric estimator for~$F_1(y\vert z):=\mathbb P(Y\leq y,D=1\vert z)$. In practice, we use nonparametric estimation procedures in \cite{LR:2008}. Lemma~\ref{lemma:CLR} (proved in the appendix) shows the applicability of the methodology in \cite{CLR:2009} under Condition~NK page~703.
\begin{lemma}
\label{lemma:CLR}
If the data generating process~$\pi$ is in~$\Pi^r$ of Definition~\ref{def:adj} and if~$G(y\vert z,\tilde z):=\mathbb P(Y\leq y\vert \tilde z)-\mathbb P(Y\leq y,D=0\vert z)$ is continuous 
in~$y$ for all~$z,\tilde z$, we have
\[
\underline C(y,z)\geq y-\inf_{\tilde z\geq z}G_-(F_1(y\vert z)\vert z,\tilde z), \; \mbox{ where } \; G_-(x\vert z,\tilde z):=\sup\{y\in\mathcal Y: G(y\vert z,\tilde z)\leq x\}.
\]
\end{lemma}

Let~$s_n(y;z,\tilde z)$ be a standard error for the estimator~$\hat G_-(\hat F_1(y\vert z)\vert z,\tilde z)$ and~$c^\alpha_n(y;z)$ be the critical value of Definition~3 in \cite{CLR:2009}. Then, under the assumptions of Theorem~6 (including Condition~NK) of \cite{CLR:2009}, 
\[
C_n(y,z):=y-\inf_{\tilde z\leq z}\left\{ \hat G_-(\hat F_1(y\vert z)\vert z,\tilde z) + c^\alpha_n(y;z)s_n(y;z,\tilde z)\right\}
\]
satisfies requirement~(\ref{eq:CR}). Condition~NK page~703 of \cite{CLR:2009} is a complex high-level assumption on the functions~$\tilde z\mapsto G_-(F_1(y\vert z)\vert z,\tilde z)$ and~$\tilde z\mapsto \hat G_-(\hat F_1(y\vert z)\vert z,\tilde z)$. Developing simpler sufficient conditions on the joint distribution of the vector~$(Y,D,Z)$ for Condition~NK to hold is beyond the scope of this paper. 

%%%%%%%%%%

% SIMULATIONS

%%%%%%%%%%

\section{Simulation study}
\label{sec:sim}

The objective of this simulation study is to illustrate and clarify the empirical content of the main assumptions in our model, namely the stochastic monotonicity assumptions~\ref{ass:MIV} and~\ref{ass:mMax}. We postulate a linear model for potential outcomes.
\begin{eqnarray*}
Y_d &=& c_d Z+(1-\sigma)v_d+\sigma\varepsilon_d,
\end{eqnarray*}
where~$c_d=0.5d+3(1-d)$, $Z$ follows a Beta$(2,5)$ distribution, $\ln (v_0,v_1)$ follows a multivariate normal distribution~N$\left(\left(\begin{array}{c}-2.25\\0\end{array}\right),\left(\begin{array}{cc}1&0.1\\0.1&1\end{array}\right)\right)$, and $\ln\varepsilon_d=2d+\varepsilon$, where~$\varepsilon$ is standard normal random variable. The choice of means reflects an earnings advantage for sector~1, and we let potential incomes be correlated conditionally on~$Z$. The shifter~$Z$ is observed by the agent and the analyst,~$(v_0,v_1)$ is observed by the agent, but not the analyst, and~$\varepsilon_d$ is observed by neither. We use~$\sigma=0$ to model the case, where agents have perfect foresight, and use~$\sigma=0.8$ otherwise. Since~$Y_d$ is an increasing function of~$Z$ plus noise, Assumptions~\ref{ass:MIV} and~\ref{ass:mMax} are satisfied.

We entertain the two parametric models for utility~(\ref{eq:QL}) and~(\ref{eq:CES}) from Section~\ref{sec:para}, with parameter values~$\alpha=1$, $\beta=0.2$ in the quasi-linear case, and~$2.5$ in the CES case, and~$\gamma=1$. The corresponding theoretical cost of choosing Sector~1 is:
\begin{eqnarray}
\label{eq:cost_QL}
C(y,z) & = & \alpha(1-z)^\beta \; \mbox{ in the QL case} \\
\label{eq:cost_CES}
& = & y - \left( y^\frac{\beta-1}{\beta} - \alpha \left( 1-z\right)^\frac{\beta-1}{\beta}\right)_+^\frac{\beta}{\beta-1}
\; \mbox{ in the CES case}.
\end{eqnarray}

Agents choose sector~$d$ that maximizes~$\mathbb E[u(Y_d,d,Z)\vert Z,v_0,v_1]$ as in Assumption~\ref{ass:if}. We consider inference on the non pecuniary cost of Sector~1 under imperfect and perfect foresight respectively.

\subsubsection*{Imperfect foresight}
We first consider the analysis of imperfect foresight agents, assuming imperfect foresight. Nonparametric inference on the costs of Sector~1 under imperfect foresight is based on the bounds in~(\ref{eq:LB}). A positive cost is detected when the conditional mean of realized incomes~$\mathbb E[Y\vert Z=z]$ is not monotonic in~$z$. The non monotonicity of realized incomes~$\mathbb E[Y\vert Z=z]$ is visualized in Figure~\ref{fig:EYZ}. The empirical content of the model is visualized on Figure~\ref{fig:IF_empirical_content}. In the latter, the theoretical cost from Equation~\ref{eq:cost_QL} is plotted together with the theoretical nonparametric lower bound~$\underline C(z)$ from Equation~(\ref{eq:LB}). The performance of the inference procedure can then be assessed in Figure~\ref{fig:IF_inference}, which plots the theoretical lower bound~$\underline C(z)$ from Equation~(\ref{eq:LB}) together with the lower envelope of the~$95\%$ confidence region for each of~$100$ samples of size~$10,000$ each.

\begin{figure}
\centering
\subfigure[QL]{
\includegraphics[scale=0.5]{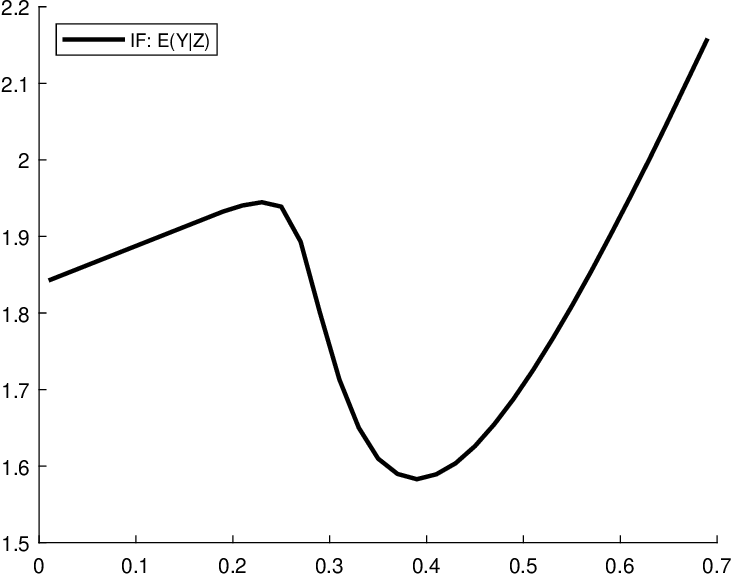}
}
\subfigure[CES]{
\includegraphics[scale=0.5]{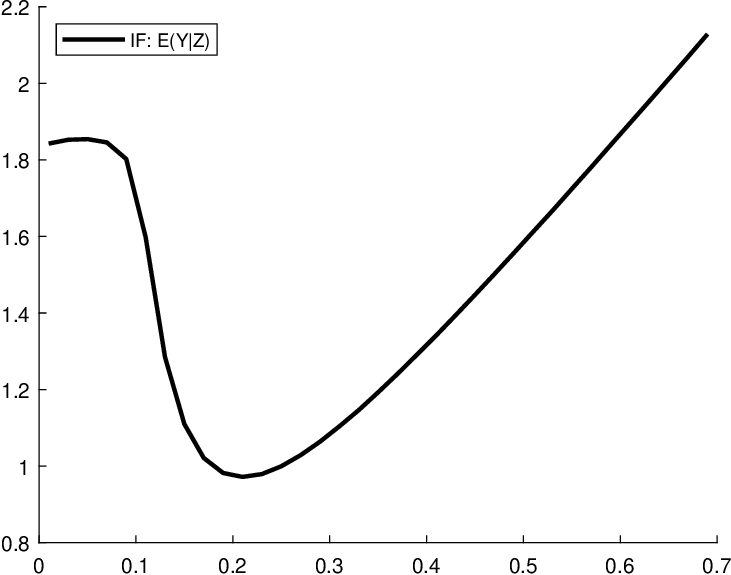}
}
\caption{Imperfect foresight: Non monotonicity of realized incomes~$\mathbb E[Y\vert Z=z]$}
\label{fig:EYZ}
\end{figure}

\begin{figure}
\centering
\subfigure[QL]{
\includegraphics[scale=0.5]{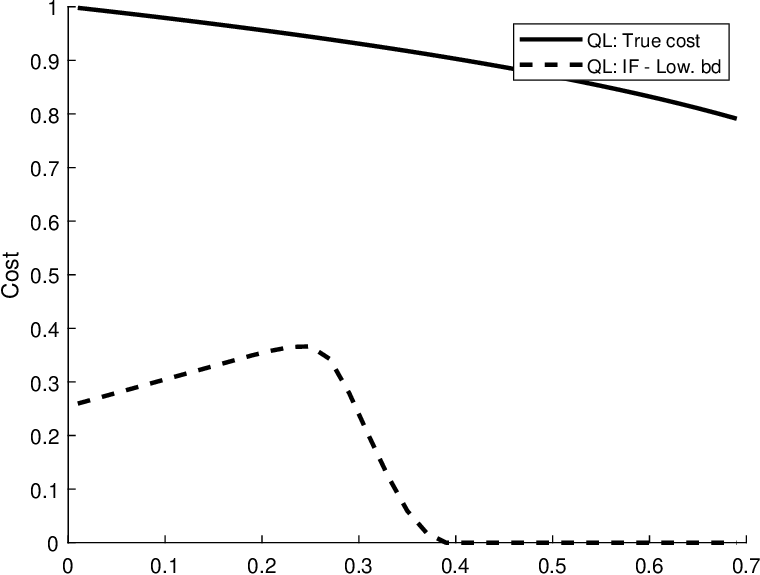}
}
\subfigure[CES]{
\includegraphics[scale=0.5]{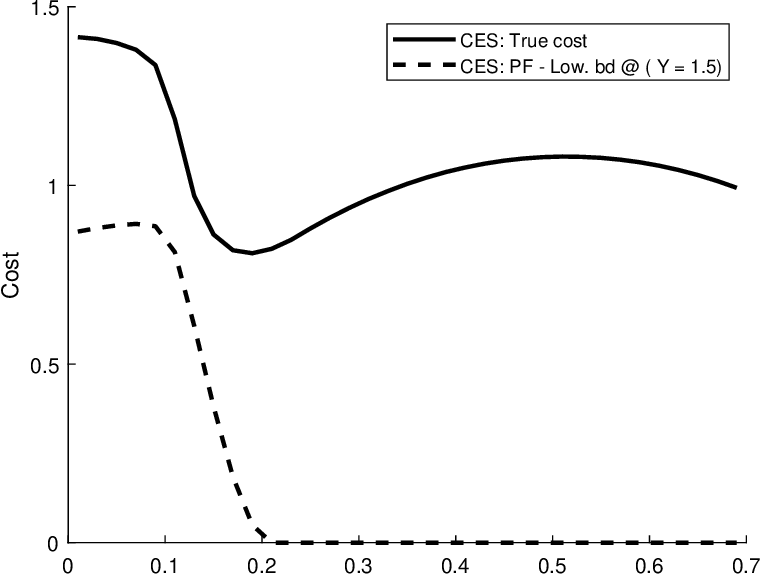}
}
\caption{Imperfect foresight: Empirical content - theoretical cost and theoretical nonparametric lower bound}
\label{fig:IF_empirical_content}
\end{figure}

\begin{figure}
\centering
\subfigure[QL]{
\includegraphics[scale=0.5]{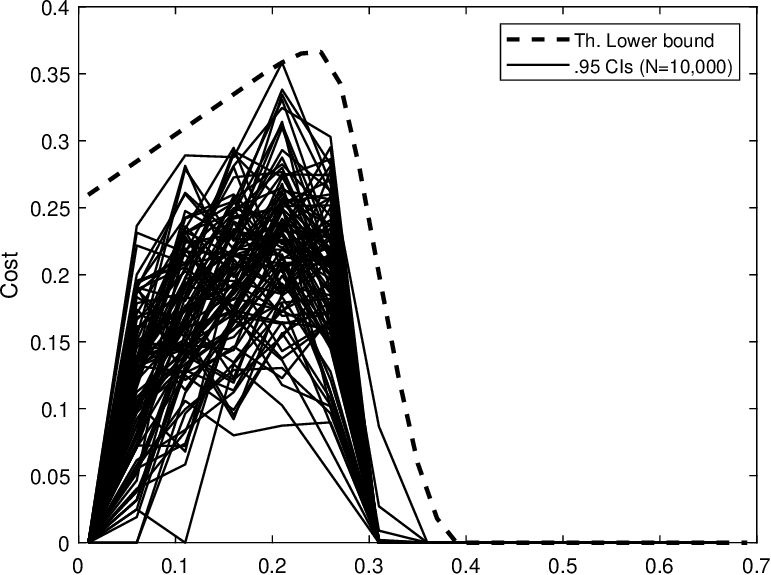}
}
\subfigure[CES]{
\includegraphics[scale=0.5]{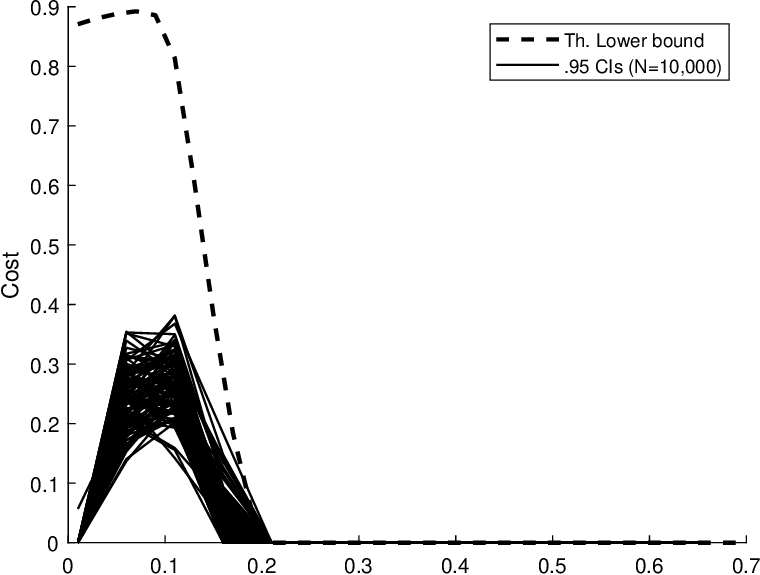}
}
\caption{Imperfect foresight: Lower envelope of the~$95\%$ confidence region}
\label{fig:IF_inference}
\end{figure}

\subsubsection*{Perfect foresight} Here we consider the analysis of perfect foresight agents, correctly assuming perfect foresight. Nonparametric inference on the costs of Sector~1 under foresight is based on the bounds in~(\ref{eq:CFB}). A positive cost is detected when the conditional cdf of realized incomes~$\mathbb P[Y\leq y\vert Z=z]$ is not monotonic in~$z$. The non monotonicity of realized incomes~$\mathbb P[Y\leq y\vert Z=z]$ is visualized in Figure~\ref{fig:FYZ}. The empirical content of the model is visualized on Figure~\ref{fig:PF_empirical_content}. In the latter, the theoretical cost from Equation~\ref{eq:cost_CES} is plotted together with the theoretical nonparametric lower bound~$\underline C(y,z)$ from Equation~(\ref{eq:CFB}). The performance of the inference procedure can then be assessed in Figure~\ref{fig:PF_inference}, which plots the theoretical lower bound~$\underline C(y,z)$ from Equation~(\ref{eq:LB}) together with the lower envelope of the~$95\%$ confidence region for each of~$100$ samples of size~$10,000$ each. Although there is empirical content, as shown in Figure~\ref{fig:PF_empirical_content}, the stochastic monotonicity test we use produces overly conservative inference, as can be seen in Figure~\ref{fig:PF_inference}.

\begin{figure}
\centering
\subfigure[QL]{
\includegraphics[scale=0.5]{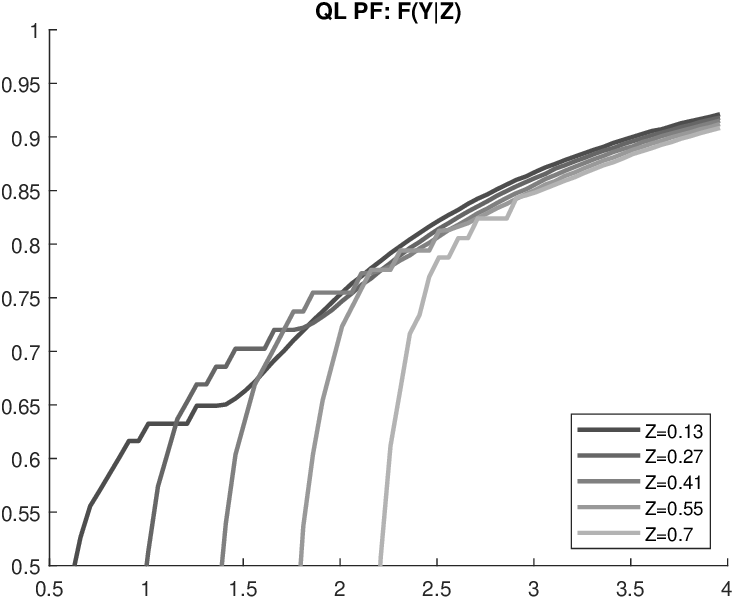}
}
\subfigure[CES]{
\includegraphics[scale=0.5]{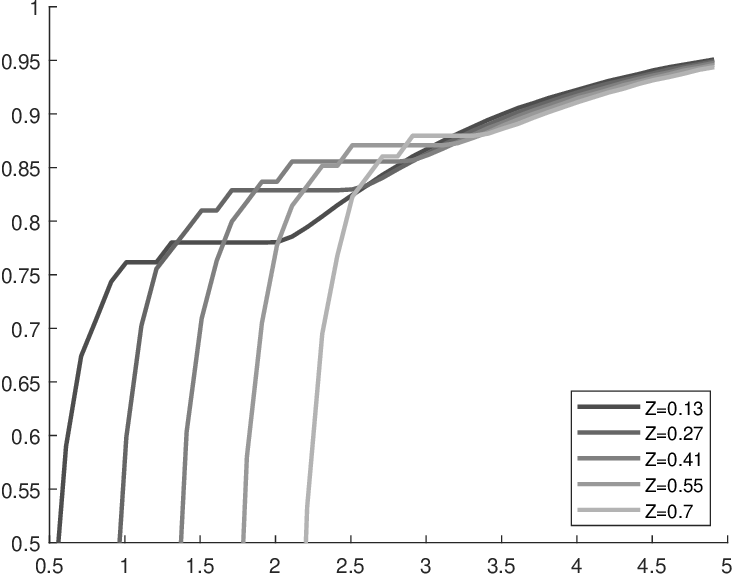}
}
\caption{Perfect foresight: Non monotonicity of realized incomes~$\mathbb P[Y\leq y\vert Z=z]$}
\label{fig:FYZ}
\end{figure}

\begin{figure}
\centering
\subfigure[QL]{
\includegraphics[scale=0.5]{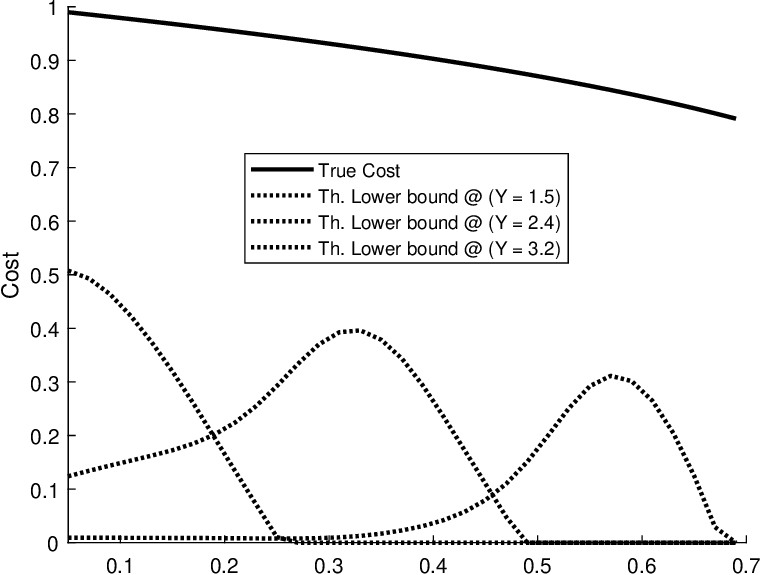}
}
\subfigure[CES]{
\includegraphics[scale=0.5]{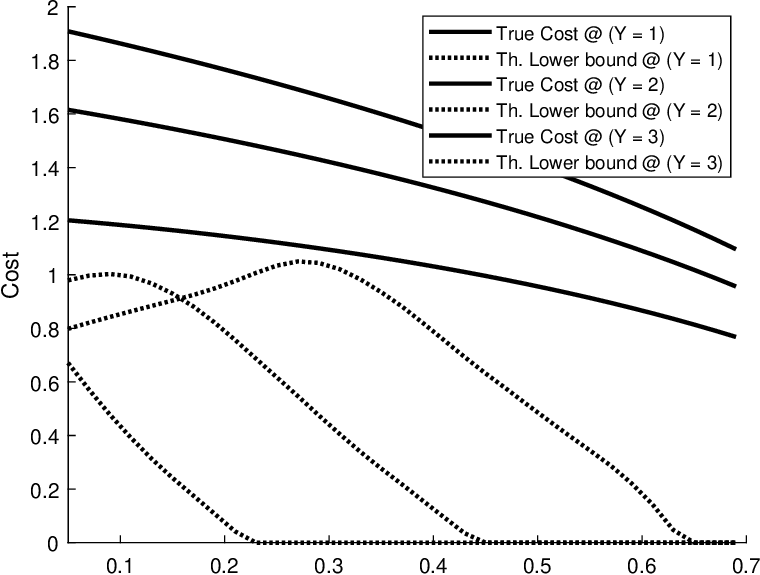}
}
\caption{Perfect foresight: Empirical content - theoretical cost and theoretical nonparametric lower bound}
\label{fig:PF_empirical_content}
\end{figure}

\begin{figure}
\centering
\subfigure[QL]{
\includegraphics[scale=0.5]{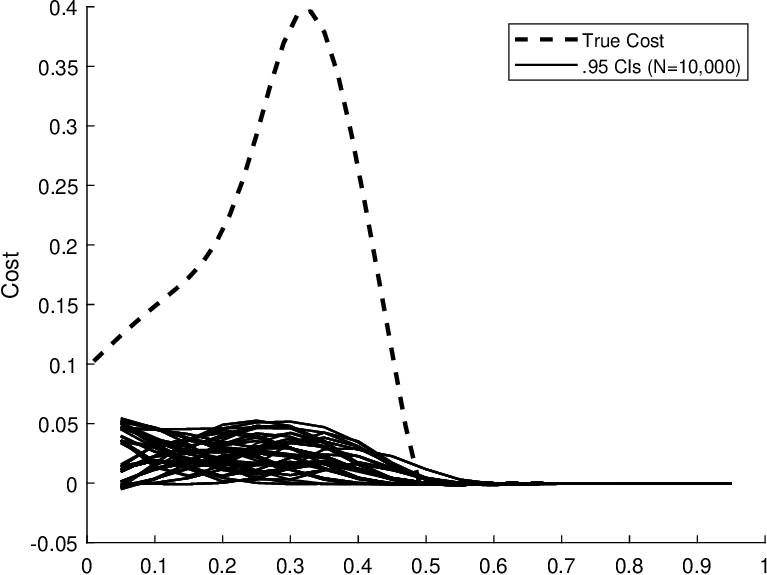}
}
\subfigure[CES]{
\includegraphics[scale=0.5]{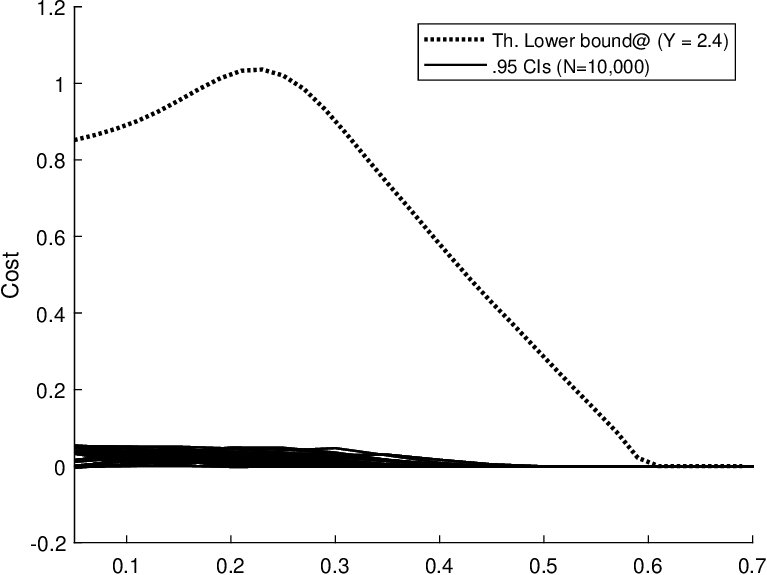}
}
\caption{Perfect foresight: Lower envelope of the~$95\%$ confidence region}
\label{fig:PF_inference}
\end{figure}

\subsubsection*{Misspecification} There is a concern that basing inference on a perfect foresight assumption may reveal spurious non pecuniary costs if agents have imperfect foresight. To evaluate this possibility, we analyze inference on the non pecuniary cost of Sector~1 for imperfect foresight quasi-linear utility agents, incorrectly assuming perfect foresight. Figure~\ref{fig:IF_misspec} shows the theoretical cost from Equation~\ref{eq:cost_QL} as a function of~$z$, together with the theoretical lower bound~$\underline C(y,z)$ from Equation~(\ref{eq:CFB}). Figure~\ref{fig:IF_misspec} shows that~$\underline C(y,z)$ can be larger than the theoretical cost, which indicates spurious cost detection due to the misspecification of the informational environment.

\begin{figure}
\centering
\subfigure[QL]{
\includegraphics[scale=0.5]{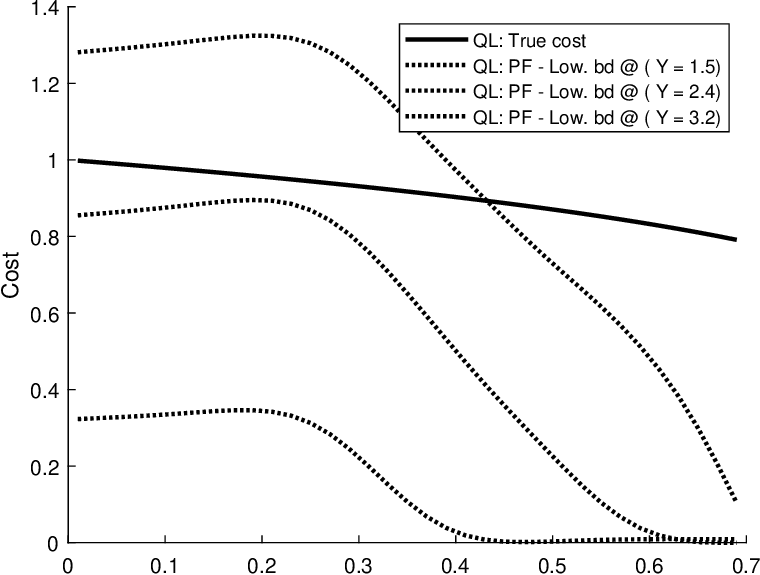}
}
\subfigure[CES]{
\includegraphics[scale=0.5]{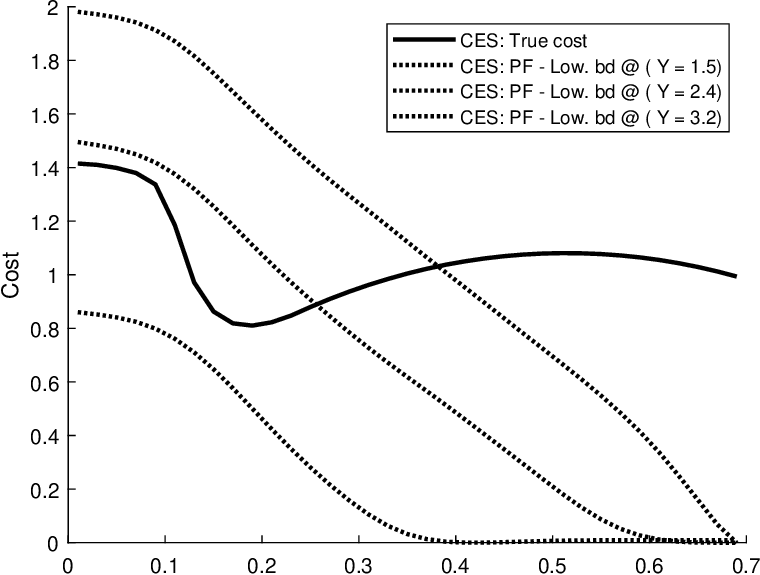}
}
\caption{Lower envelope of the~$95\%$ confidence region}
\label{fig:IF_misspec}
\end{figure}

%%%%%%%%%%%%%%%%%%%%%%%%%%%%%%%%%%%%%%%%%%%%

%                                 MAJOR CHOICE

%%%%%%%%%%%%%%%%%%%%%%%%%%%%%%%%%%%%%%%%%%%%

\section{Empirical illustration: Major choices in Germany}
\label{sec:appli}

There is ample evidence that women are severely under-represented in STEM university majors and even more so in STEM fields (see for instance \cite{beede2011}, \cite{zafar2013} and \cite{HGHM:2013}). Evidence on women's reasons for shunning STEM fields (see for instance the survey in \cite{KG:2017}) include mathematics gender stereotypes and gender biased amenities, such as family friendliness and work/life balance. Our objective is to document the amplitude of such non pecuniary motivations as revealed in the form of a gender-specific cost of choosing STEM fields. The revealed cost of STEM fields is a function of the rate of feminization of the STEM faculty in the region at the time of choice. We therefore also shed light on the importance of role models in the determination of major choice (\cite{KG:2017}).

% DATA
 
\subsection{Data}
Our empirical analysis relies on surveys of German nationally representative university graduates. The data are collected by the German Centre for Higher Education Research and Science Studies (DZHW) as part of the DZHW Graduate Survey Series. Data and methodology are described in \cite{DZHW}. The waves we consider include graduates who obtained their highest degree during the academic years 2008-2009.
Graduates were interviewed~1 year and~5 years after graduation. At that point, extensive information was collected on their educational experience, employment history, including wages and hours worked, along with detailed socio-economic variables and geographical information about the region where the {\em Abitur} (high school final exam) was completed.
We merge the fields of study into two categories. We call STEM the category, which consists of mathematics, physical, life and computer sciences, as well as engineering and related fields.
The remaining majors are merged in the non-STEM-degree category. We only consider graduates from institutions in the country of the survey, who are active on their respective country's labor market at the time of the interview. 

Our proposed stochastically monotone instrumental variables (SMIV) are the mother's education level and a variable, which we call ``feminization of STEM.'' The latter is defined as the proportion of women among faculty members by field of study in universities in the individual's region (Land) of residence at the time of choice, i.e., the German Land in which the individual graduated from High School, not the Land where the individual attends university. 
This variable is calculated for each individual in the sample from data on gender distribution of faculty by field and by Land provided by the Federal Statistical Office of Germany (DESTATIS). The data set provides for each year between 1998 and 2010, the count of faculty members (Scientific and artistic staff ``{\em Wissenschaftliches und Künstlerisches Personal}'') by gender in ten fields of study. The variable is computed from the aggregation of mathematics, science and engineering. 

The validity of the stochastically monotone instrumental variable rests on a combination of factors. A possible channel is the role model effect: Female students perform better in an environment with more role models. Another channel works through the provision of female specific amenities: Having more women on the faculty is likely to increase the provision of female specific amenities, which will increase women's utility in the sector with a larger proportion of women on the faculty. There is also a selection effect at work, but the latter is more ambiguous. Because of selection, female faculty in STEM may be drawn from a better selected distribution of talent than their male counterparts. Now, more talented women are likely to be more effective role models. However, the selection process we describe also means that if the overall talent distribution is the same in all locations, locations with fewer women on the faculty in STEM would also have a better selected talent distribution. Hence, the effect on role model effectiveness would be the ambiguous result of a trade-off between quality and quantity.

A more serious concern we have (as discussed in our conclusion) is the effect of aggregation into two sectors. The reason it may sometimes lead to violations of monotonicity of potential wages is that within STEM, some sub sectors may have higher wages and lower feminization. This concern is alleviated somewhat by the fact that we do not require wages to be stochastically monotone. The stochastic monotonicity assumption requires that the vector of potential utilities (functions of potential wages and amenities, including feminization) be stochastically monotone with respect to feminization of university faculty. This could be satisfied even in certain cases, where the potential wages themselves are not stochastically monotone.

\subsubsection*{Response rate and attrition}

Several sources of sample selection are of concern. However, given the scarce socio-economic information we have about individuals in the sample, we cannot correct for potential sample selection using matching with external data sets.
First, the response rate was~$20\%$ for the~2009 cohort with~$14\%$ attrition in the second wave.
Second, we exclude women who are still enrolled in higher education ($7.8\%$ of the sample), women who are part-time, self-employed ($6.9\%$ of the sample), and women who are unemployed or out of the labor market ($2\%$ of the sample). We keep only graduates who hold a ``Bachelor'', ``Magister'' or ``Diplom'', excluding those with ``Staatsexamen'' and ``Lehramt'' degrees, which are specific tracks mainly for teachers ($22\%$ of the sample). 
Finally, \cite{Heublein:2014} estimates the drop-out rate for German bachelors degrees to be between $28$ and~$30\%$ during the period of interest. The dropout rate is higher for STEM fields, where the estimates range from~$30$ to~$39\%$. There is some evidence of higher female drop-out rates in other contexts, see \cite{Saltiel:2018} and references therein. However, we were unable to locate a reference with estimates of gender-specific drop-out rates in German universities during the period of interest.

% DESCRIPTIVE STATISTICS

\subsection{Descriptive statistics}
The category of individuals we consider consists of women from the former West Germany\footnote{For the sample of women from the former East Germany, the results (omitted here for space constraints) reveal much lower costs of STEM. The difference in behavior between the former East and West may be partly attributable to differences in gender stereotypes, as evidenced in \cite{LS:2018}}. The variables we consider are average income during the fifth year after graduation, which serves as outcome variable~$Y$,\footnote{This choice is motivated by the available data. Lifetime earnings or wage growth opportunities may be important determinants of choice as well, as suggested by the findings in \cite{MO:2020}. We are grateful to an anonymous referee for bringing the latter reference to our attention.} the choice of major~$D$, which takes value~$1$ if the chosen sector is STEM and~$0$ otherwise, and the feminization of STEM, which serves as our SMIV instrumental variable~$Z$. 
\begin{table}[htbp]
  \centering
  \caption{Sample of women and comparison with men}
    \begin{tabular}{lrr}
    \hline \hline \\
          & Women & Men  \\
\hline \\
STEM    &  299   & 610 \\
Other    & 1,332 & 471 \\  \hline  \hline
    \end{tabular}%
  \label{tab:sample}%
\end{table}%
Table~\ref{tab:sample} gives numbers of STEM majors among women from the former West Germany in 2009, as compared to proportions of STEM majors among men of the same category. 
Figure~\ref{fig:feminisationpdf} shows the distribution of the feminization of STEM variable and the mother's education for cohorts graduating in~$2009$. 
The relation between income, field of study and mother's education or feminization of STEM is investigated with an estimation of the propensity score (probability of choosing STEM) as a function of the feminization of STEM and the mother's education in~2009 in Figure~\ref{fig:income_smiv} and quartile regressions of income as a function of the feminization of STEM and the mother's education for~2009 in Figure~\ref{fig:relation_ydz}. 
Figure~\ref{fig:income_smiv} shows a positive relationship between major choice and feminization of STEM (which may or may not be causal).  

\begin{figure}[hbtp]
\caption{Density of the SMIVs}
\vskip10pt
\centering
\subfigure[Feminization of STEM]{
\includegraphics[scale=0.5]{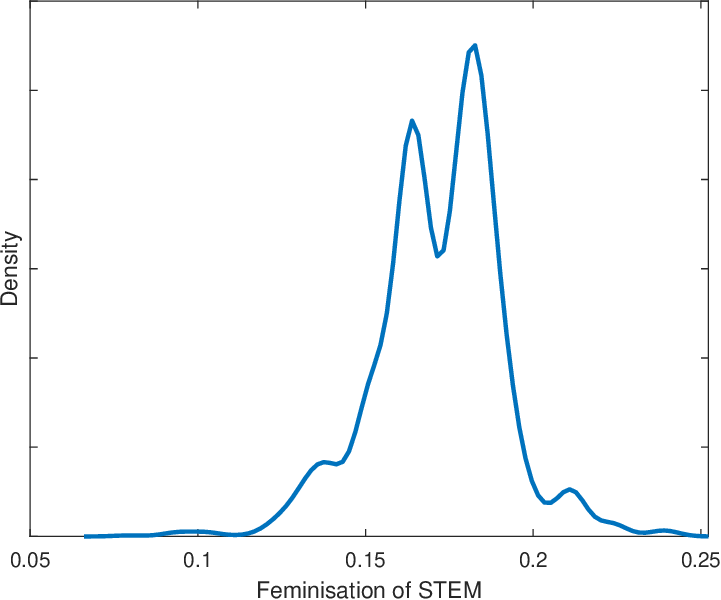}
}
\subfigure[Mother's education]{
\includegraphics[scale=0.5]{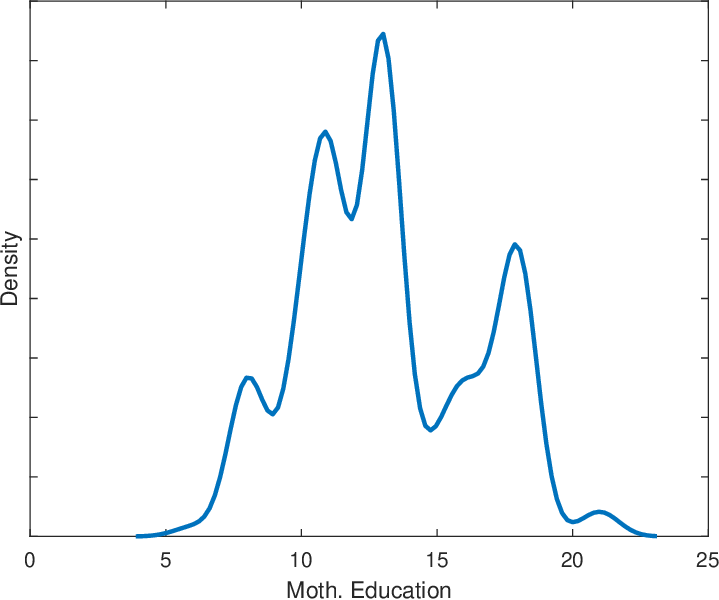}
}
\label{fig:feminisationpdf}
\end{figure}

Note that Figure~\ref{fig:relation_ydz} suggests a potential  violation of stochastic monotonicity of income relative to the feminization of STEM: Stochastic monotonicity of income relative to feminization would imply that the conditional probability of having an income below a given value does not increase with feminization~$(P(Y\leq y\vert Z=z ) \searrow z)$. This is clearly not the case in Figure~\ref{fig:relation_ydz}. This implies a violation of the pure Roy model (without non pecuniary cost) as shown in \cite{MHM:2020}. The violation is confirmed by the more formal test of the  pure Roy model entertained  in \cite{MHM:2020} for this category of individuals\footnote{Note that because of selection, a violation of stochastic monotonicity of realized income relative to feminization does not necessarily imply a violation of stochastic monotonicity of the vector of potential incomes relative to feminization.}.

\begin{figure}[hbtp]
\caption{Relation between the propensity to choose STEM and the SMIV}
\vskip10pt
\centering
\subfigure[Feminization STEM]{
\includegraphics[scale=0.5]{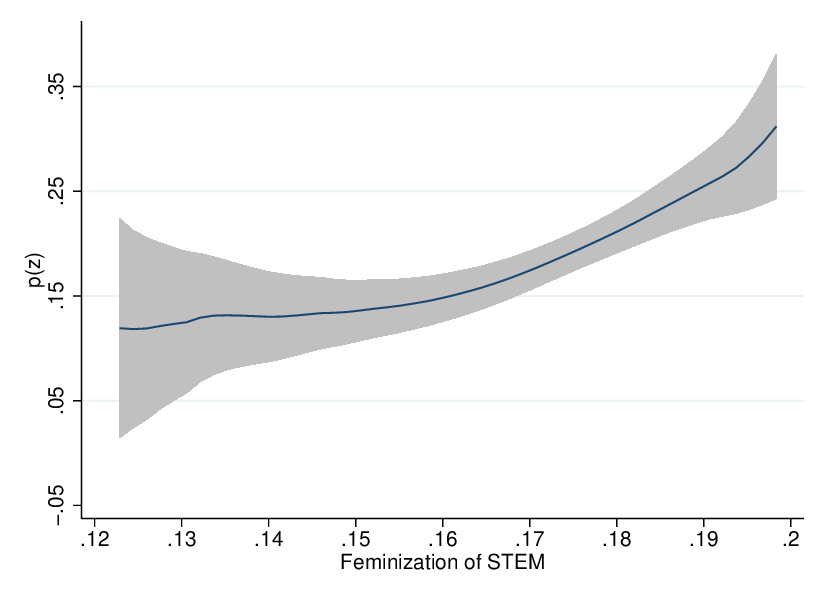} 
}
\subfigure[Mother's education]{
\includegraphics[scale=0.5]{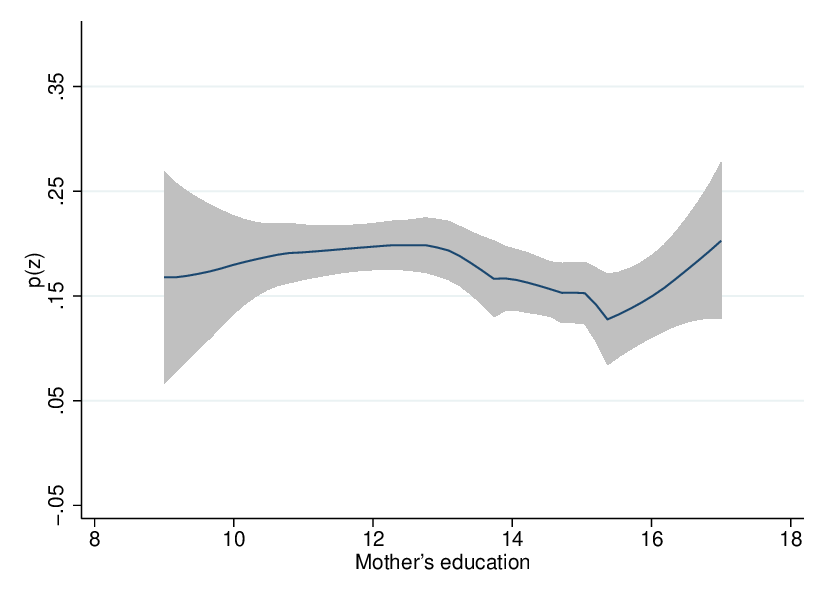} 
}
\label{fig:income_smiv}
\end{figure}

\begin{figure}[hbtp]
\vskip10pt
\caption{Relation between income distribution and the SMIV}
\centering
\subfigure[Feminization STEM]{
\includegraphics[scale=0.5]{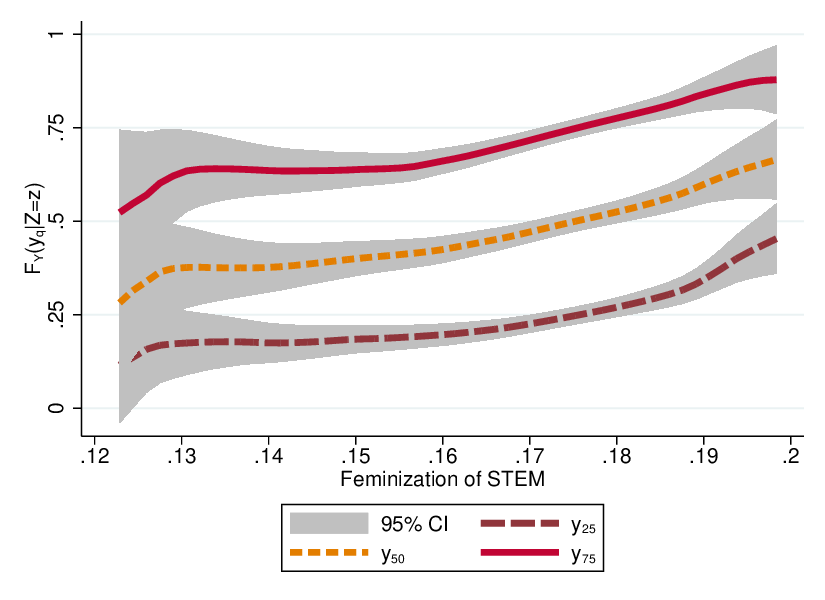} 
}
\subfigure[Mother's education]{
\includegraphics[scale=0.5]{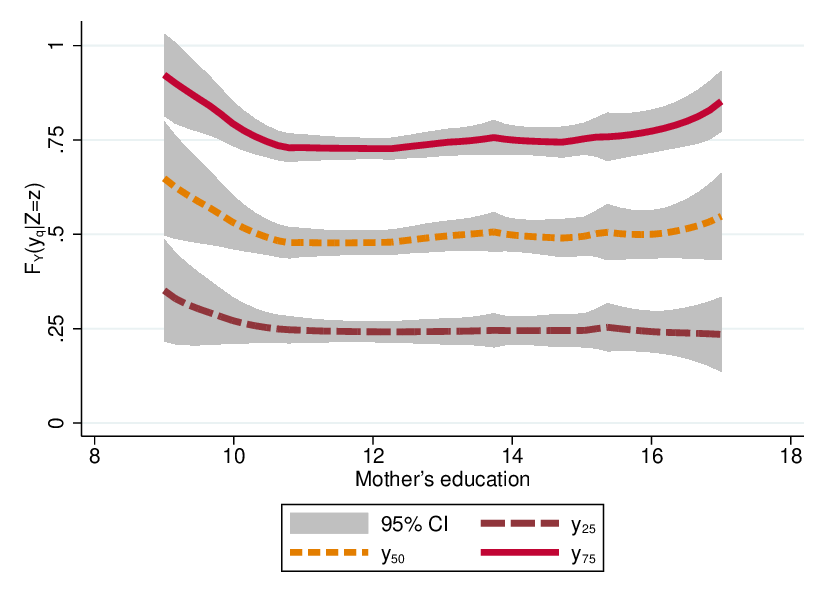} 
}
\label{fig:relation_ydz}
\end{figure}

% FINDINGS

\subsection{Findings}

Figure~\ref{fig:cr_income} shows the lower bound of a one-sided~95\% confidence region for the cost function, in the case of perfect foresight. More precisely, the left panel represents a function~$C_n(y,z)$ of income and feminization of STEM, such that~$\lim_{n\rightarrow\infty}\mathbb P(C_n(y,z)\leq \underline C(y,z))\geq0.95$, where~$\underline C$ is the lower bound of~(\ref{eq:CFB}). The right panel of Figure~\ref{fig:cr_income} shows the same, except that the feminization of STEM is replaced with the mother's education level.
We observe that costs tend to be high for low income individuals and remain high for high income individuals, when the rate of feminization of STEM is low. 
Figure~\ref{fig:surv_costdivincome} shows a different visualization the data from Figure~\ref{fig:cr_income}. 
Figure~\ref{fig:surv_costdivincome} shows, for each share~$c_0$ of total income,  the proportion of individuals for whom the upper (resp. lower) bounds of a 95\% confidence region for the cost as a share of income is greater than~$c_0$. 
Panels~(a) and~(b) show this for the feminization of STEM and the mother's education, respectively. 

\begin{figure}[hbtp]
\caption{Lower bound of the 95\% confidence region for the cost of STEM~$C(y,z)$ in perfect foresight.}
\vskip10pt
\centering
\subfigure[Feminization of STEM]{
\includegraphics[scale=0.48]{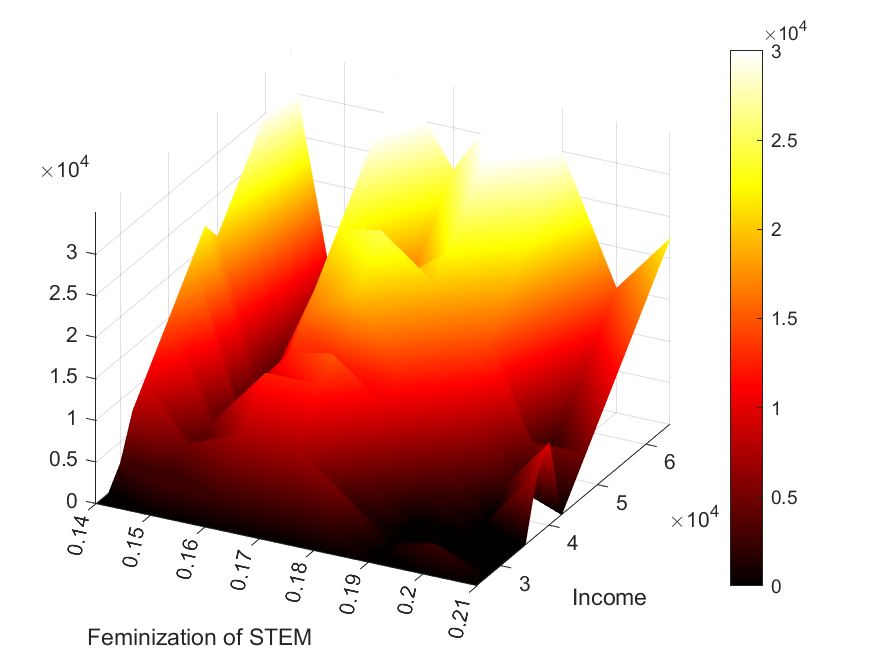}
}
\subfigure[Mother's education]{
\includegraphics[scale=0.48]{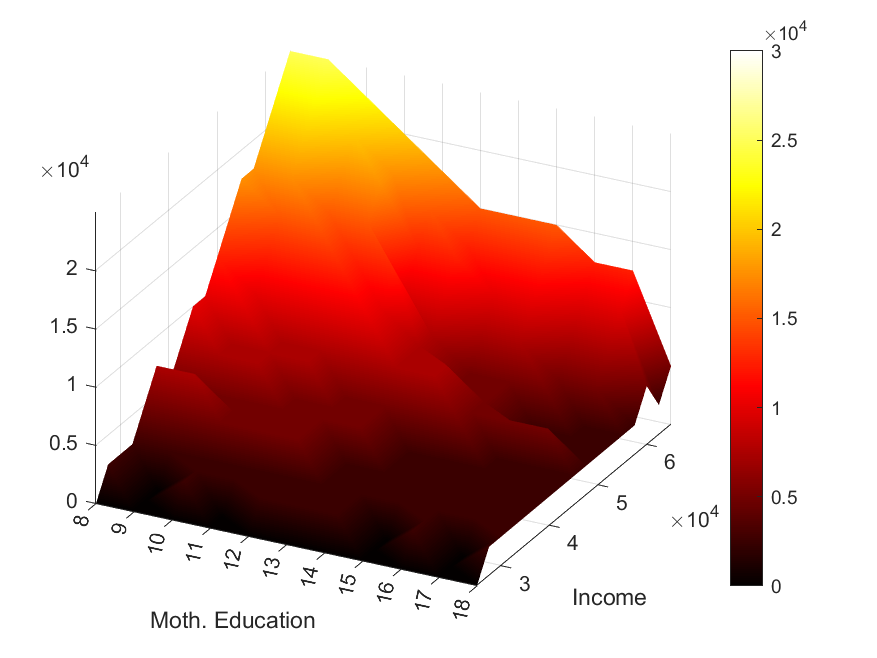}
}
\label{fig:cr_income}
\end{figure}

\begin{figure}[hbtp]
\caption{The upper blue (resp. lower red) curve traces the proportion of the population for whom the lower (resp. upper) bound of the 95\% confidence region for the cost of STEM as a proportion of income is below~$c_0$.}
\centering
\vskip10pt
\subfigure[Feminization of STEM]{
\includegraphics[scale=0.5]{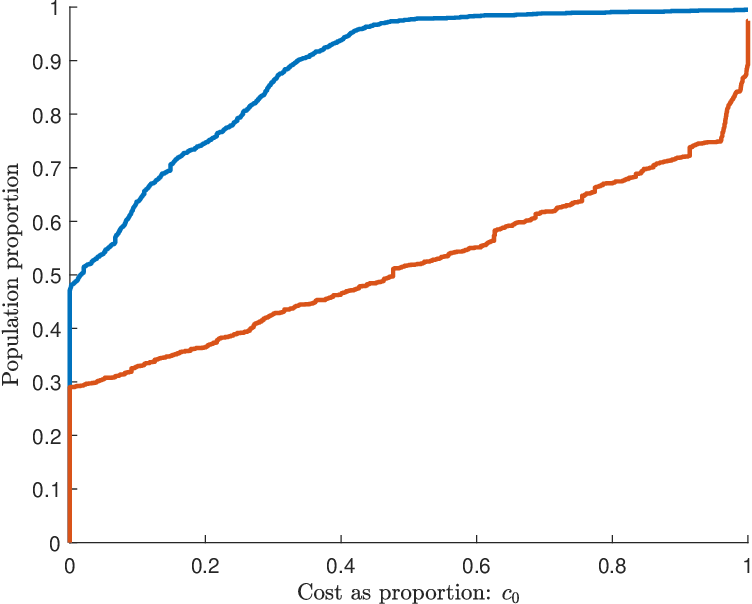}
}
\subfigure[Mother's education]{
\includegraphics[scale=0.5]{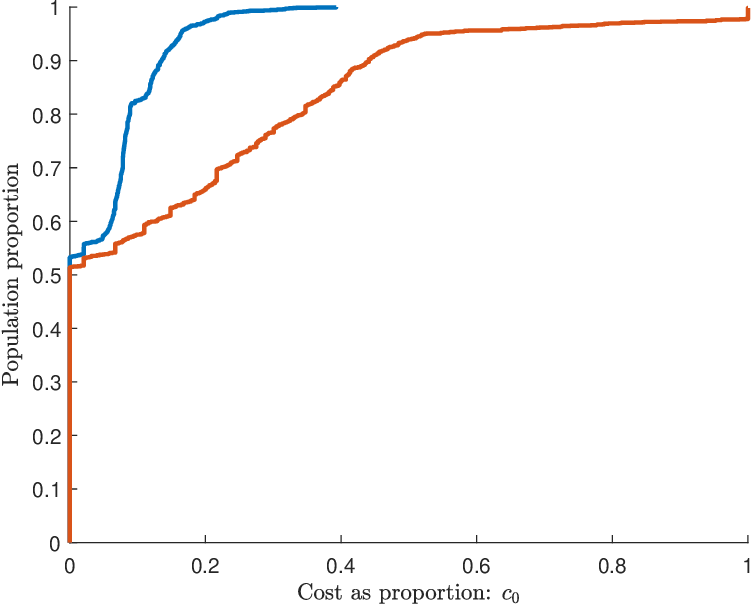}
}
\label{fig:surv_costdivincome}
\end{figure}

Figure~\ref{fig:cr_income_IF} shows the estimated cost and lower bound of the 95\% confidence region 
for the non pecuniary cost of STEM based on the imperfect foresight closed form expressions in Corollary~\ref{cor:cf}. The lower bounds include~$0$, hence are uninformative. It is expected from the theory that perfect foresight bounds be tighter than the imperfect foresight ones. However, the sharp contrast we observe here between the informativeness of the perfect foresight bounds and the non informativeness of the imperfect foresight ones is a cause for concern. As we have seen in Section~\ref{sec:sim}, there is a possibility of spurious cost detection as a result of misspecification of the perfect forecast model.

\begin{figure}[hbtp]
\caption{Estimated lower bound (red) and lower bound of the 95\% confidence region (blue) for the cost of STEM in imperfect foresight.}
\vskip10pt
\centering
\subfigure[Feminization of STEM]{
\includegraphics[scale=0.48]{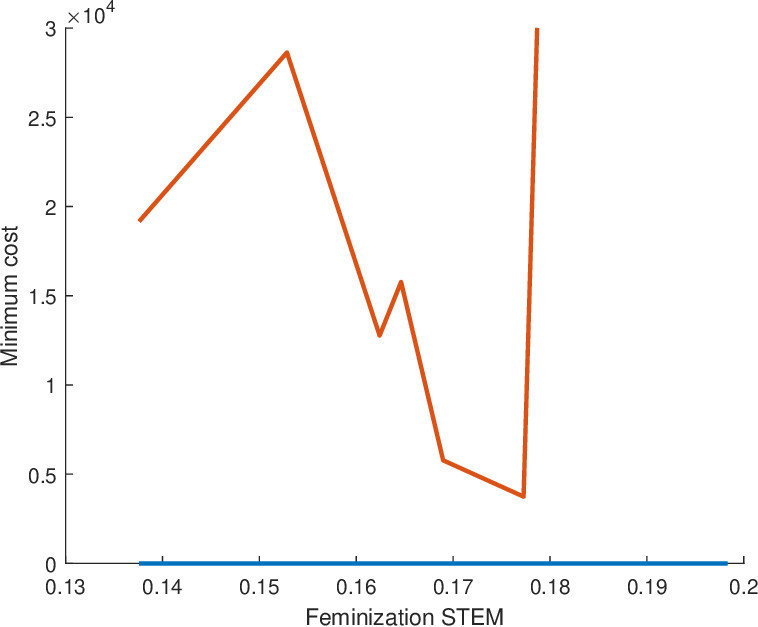}
}
\subfigure[Mother's education]{
\includegraphics[scale=0.48]{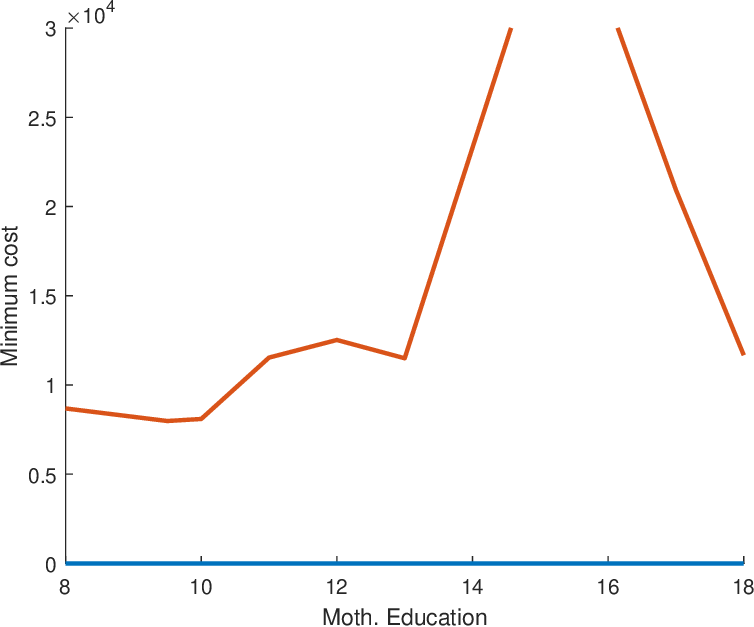}
}
\label{fig:cr_income_IF}
\end{figure}

%%%%%%%%%%%%%%%%%%%%%%%%%%%%%%%%%%%%%%%%%%%%

%                                 CONCLUSION

%%%%%%%%%%%%%%%%%%%%%%%%%%%%%%%%%%%%%%%%%%%%

\section{Discussion}

We combine an extended Roy model with a stochastic monotonicity constraint. We are motivated by the need to uncover causes of under representation of women in STEM higher education. On the one hand, it is well documented that expected potential earnings are not the only drivers of major choice, and that non pecuniary, sometimes gender specific factors affect decisions. This motivates the extended Roy model of major choice, where students choose the major that maximizes a utility that depends on both earnings and amenities. On the other hand, there is a sizable literature on the positive effects of role models on choice and outcomes. This motivates the assumption that the vector of expected earnings is stochastically monotone with respect to a variable that measures the influence of role models, namely the mother's education and the proportion of faculty in STEM that are female. We derive testable implications of the model that combines the extended Roy selection with the stochastic monotonicity constraint. We also derive closed form sharp bounds on the implied non pecuniary cost of STEM for women.

Our current methodology cannot disentangle the role of real or perceived gender biased disamenities of STEM fields, in terms of family friendliness and work/life balance, from behavioral and preference biases related to gender stereotypes. An area of concern regarding the validity of our identifying stochastic monotonicity assumption is the aggregation in the STEM category, of areas such as engineering, with extremely low feminization rates and high relative incomes, with areas such as life sciences, with high feminization rates and low relative incomes. Access to more disaggregated data on the proportions of female faculty in mathematics intensive fields (rather than STEM fields as traditionally classified) would considerably alleviate this concern.

%%%%%%%%%%%%%%%%%%%%%%%%%%%%%%%%%%%%%%%

%                         APPENDIX    

%%%%%%%%%%%%%%%%%%%%%%%%%%%%%%%%%%%%%%%

\begin{appendix}

{\scriptsize

% STOCHASTIC MONOTONICITY

\section{Stochastic monotonicity}
\label{app:SM}

We rely on the definition of first order stochastic dominance for random vectors.

\begin{definition}[Upper Sets]\label{def:US}
A set $U\subseteq\mathbb R^k$ is called an {\em upper set} if $y\in U$ implies $\tilde y\in U$ for all~$\tilde y\geq y$.
\end{definition}

\begin{definition}[First Order Stochastic Dominance]\label{def:SS}
A random vector $X_1$ is first order stochastically dominated by a random vector~$X_2$ if and only if $\mathbb P(X_1\in U)\leq\mathbb P(X_2\in U)$ for all upper sets~$U$ in~$\mathbb R^k$.
\end{definition}

\begin{definition}[Stochastic Monotonicity]\label{def:S-Mon}
A random vector $X$ on $\mathbb R^k$ is said to be {\em stochastically monotone} (non decreasing) with respect to a random vector~$Z$ on~$\mathbb R^l$ if the conditional distribution of~$X$ given~$Z=z_1$ is first order stochastically dominated by the conditional distribution of $X$ given $Z=z_2$, for any pair $(z_1,z_2)$ of elements of the support of $Z$ such that~$z_2\geq z_1$.
\end{definition}

% PROOFS

\section{Proofs of results in the main text}

\begin{proof}[Proof of Theorem~\ref{thm:charIF}]
Fix any $\pi\in\Pi$, and any $u\in\mathcal U(\pi)$. By Definition~\ref{def:idsetIF}, there exists a random vector $(Y_0,Y_1,D,Z)$ where $((1-D)Y_0+DY_1,D,Z)$ has distribution~$\pi$ and Assumptions~\ref{ass:if} and~\ref{ass:MIV} are satisfied. Call $Y:=(1-D)Y_0+DY_1$. By Assumption~\ref{ass:if}, if $D=d,$ then $\mathbb E[u(Y,D,Z)\vert \mathcal I]=\mathbb E[u(Y_d,d,Z)\vert \mathcal I]\geq\mathbb E[u(Y_{1-d},1-d,Z)\vert \mathcal I]$, for~$d=0,1$. Hence, $\mathbb E[u(Y,D,Z)\vert \mathcal I]=\max\{\mathbb E[u(Y_0,0,Z)\vert \mathcal I],\mathbb E[u(Y_1,1,Z)\vert \mathcal I]\}$. By Assumption~\ref{ass:MIV}, $\mathbb E[u(Y,D,Z)\vert \mathcal I]$ is therefore stochastically monotone with respect to~$Z$. Since~$Z$ is $\mathcal I$-measurable, iterated expectations yields monotonicity of the function~$z\mapsto\mathbb E[u(Y,D,Z)\vert Z=z]$ as required.

Conversely, suppose the function~$u$ on~$\mathcal Y\times\{0,1\}\times\mathcal Z$ is continuous and increasing in its first argument and is such that $\mathbb E[u(Y,D,Z)\vert Z=z]$ is monotonic non decreasing in~$z$. Then, if we define\footnote{This is where the condition~$\mathbb E(u(Y,D,Z)\vert Z)\geq\mathbb E(u(\underline b,1-D,Z)\vert Z)$ is required.} $Y_0:=u^{-1}(u(Y,D,Z),0,Z)$ and $Y_1:=u^{-1}(u(Y,D,Z),1,Z)$, the vector $(u(Y_0,0,Z),u(Y_1,1,Z))$ is equal to $(u(Y,D,Z),u(Y,D,Z))$, and therefore satisfies Assumption~\ref{ass:MIV}. Assumption~\ref{ass:if} is also trivially satisfied, since it places non constraint on~$D$ in case of a tie. Finally, $(1-D)Y_0+DY_1=Y$. Hence, there exists a random vector $(Y_0,Y_1,D,Z)$ where $((1-D)Y_0+DY_1,D,Z)$ has distribution~$\pi$ and Assumptions~\ref{ass:if} and~\ref{ass:MIV} are satisfied, and $u$ therefore belongs to the identified set~$\mathcal U(\pi)$.
\end{proof}

\begin{proof}[Proof of Corollary~\ref{cor:cf}]
By Theorem~\ref{thm:charIF}, $\mathbb E[u(Y,D,Z)\vert Z=z]$ is non decreasing in~$z$. Hence, by the definitions of~$C$ and~$\tilde{\mathcal Z}$, $\mathbb E[Y-DC(Z)\vert Z=z]=\mathbb E[Y\vert Z=z]-p(z)C(z)$ is non decreasing on~$\tilde{\mathcal Z}$. Hence 
\[
\mathbb E[Y\vert Z=z]-p(z)C(z)\leq \inf_{\tilde z\geq z,\,z\in\tilde{\mathcal Z}}\mathbb E[Y\vert Z=\tilde z].
\]
So $C(z)\geq \underline C(z),$ by the definition of~$\underline C$ in Equation~(\ref{eq:LB}). The upper bound is treated identically.
\end{proof}

\begin{proof}[Proof of Theorem~\ref{thm:char}]
Fix any $\pi\in\Pi$, and any $u\in\mathcal U^\prime(\pi)$. By Definition~\ref{def:idset}, there exists a random vector $(Y_0,Y_1,D,Z)$ where $((1-D)Y_0+DY_1,D,Z)$ has distribution~$\pi$ and Assumptions~\ref{ass:PO}, \ref{ass:sel}, and~\ref{ass:mMax} are satisfied. Call $Y:=(1-D)Y_0+DY_1$. For any $y\in\mathcal Y$, by Assumption~\ref{ass:sel}, $u(Y,D,Z)\leq y$ if and only if $u(Y_0,0,Z)\leq y$ and $u(Y_1,1,Z)\leq y$, so that $\mathbb P(u(Y,D,Z)\leq y\vert Z)=\mathbb P(u(Y_0,0,Z)\leq y,u(Y_1,1,Z)\leq y\vert Z)$. By Assumption~\ref{ass:mMax}, the latter is monotone non increasing in~$z$. Hence, for any $y\in\mathcal Y$, $\mathbb P(u(Y,D,Z)> y\vert Z=z)$ is non decreasing in~$z$, so for any ordered pair $z\geq \tilde z,$ of the support of~$Z$, (\ref{eq:char}) holds. 

Conversely, suppose the function~$u$ on~$\mathcal Y\times\{0,1\}\times\mathcal Z$ satisfies~(\ref{eq:char}) for all ordered pairs~$z\geq \tilde z$ of the support of~$Z$. Then, if we define $Y_0:=u^{-1}(u(Y,D,Z),0,Z)$ and $Y_1:=u^{-1}(u(Y,D,Z),1,Z)$, the vector $(u(Y_0,0,Z),u(Y_1,1,Z))$ is equal to $(u(Y,D,Z),u(Y,D,Z))$, which is stochastically monotone with respect to~$Z$. This implies that Assumption~\ref{ass:mMax} holds for such a pair~$(Y_0,Y_1)$. Since, by construction, $u(Y_0,0,Z)=u(Y_1,1,Z)=u(Y,D,Z)$, Assumption~\ref{ass:sel} is also trivially satisfied, since it places no constraint on~$D$ in case of a tie. Finally, $(1-D)Y_0+DY_1=Y$. Indeed, when~$D=0$, we have~$Y_0=u^{-1}(u(Y,0,Z),0,Z)=Y$, and when~$D=1$, we have~$Y_1=u^{-1}(u(Y,1,Z),1,Z)=Y$ as required. Hence, there exists a random vector $(Y_0,Y_1,D,Z)$ where $((1-D)Y_0+DY_1,D,Z)$ has distribution~$\pi$ and Assumptions~\ref{ass:sel}, \ref{ass:mMax} are satisfied, and $u$ therefore belongs to the identified set~$\mathcal U^\prime(\pi)$.
\end{proof}

\begin{proof}[Proof of Corollary~\ref{cor:pointwise}]
From Equation~(\ref{eq:pointwise}) proved in the main text for all~$y\in\mathcal Y$, we get $L(y-C(y,z)\vert z) \leq \mathbb P(Y-C(Y,Z)\leq y-C(y,z), D=1\vert z) \leq U(y-C(y,z)\vert z)$.
By the definition of~$U^-$ and since~$y\mapsto y-C(y,z)$ is increasing, the second inequality implies~$y-C(y,z)\geq U^-(F_1(y\vert z)\vert z)$ as desired.
Similarly, by the definition of~$L_-$ and since~$y\mapsto y-C(y,z)$ is increasing, the first inequality implies~$y-C(y,z)\leq L_-(F_1(y\vert z)\vert z)$.
\end{proof}

\begin{proof}[Proof of Lemma~\ref{lemma:env}]
We prove the results relating to~$\underline F$. The results relating to~$\bar F$ are treated symmetrically. Let $\mathcal M$ be the set of non decreasing real valued functions from~$\mathbb R$ to~$[0,1]$. Define the operator~$S$ by
\begin{eqnarray*}
\begin{array}{cllc}
S: & 2^{\mathcal M} & \rightarrow & \mathcal M \\
& K & \mapsto & S(K),
\end{array}
\end{eqnarray*}
where $S(K)$ is defined for each~$y\in\mathbb R$ by $S(K)(y)=\sup\{F(y); F\in K\}$. Note that $S(K)$ is uniquely defined and belongs to~$\mathcal M$ as required.
Next, define the operator~$R$ by 
\begin{eqnarray*}
\begin{array}{cllc}
R: & \mathcal M & \rightarrow & \mathcal M \\
& F & \mapsto & R(F),
\end{array}
\end{eqnarray*}
where $R(F)$ is defined for each~$y\in\mathbb R$ by $R(F)(y)=\lim_{\tilde y\downarrow y}F(\tilde y)$. Note that~$R(F)$ is uniquely defined and belongs to~$\mathcal M$ as required.
In addition, $R$ is idempotent: $R(F)$ is right-continuous by construction, so that~$R(R(F))=R(F)$, and monotone: if $F_1\leq F_2$, then $R(F_1)\leq R(F_2)$.

Define~$F_K:=R(S(K))$. We first show that $F_K$ is a cdf if $K$ is a collection of cdfs bounded above by the cdf~$1\{\cdot\geq\underline b\}$. Indeed, $F_K$ is non decreasing by construction of~$S$ and~$T$, $F_K$ is right-continuous because~$R$ is idempotent, $F_K$ tends to~$1$ at~$+\infty$ since each element of~$K$ is a cdf. Finally, $F_K$ tends to~$0$ at $-\infty$ since $0\leq F_K\leq R(\bar F).$ Define~$K_z:=\{F(\cdot\vert \tilde z): \tilde z\geq z,\tilde z\in\tilde{\mathcal Z}^\prime\}$. Then,~$\underline F$ of Lemma~\ref{lemma:env} is equal for each~$z\in\tilde{\mathcal Z}^\prime$ to $\underline F(\cdot\vert z)=R(S(K_z))$. By construction, $S(K_z)\leq S(K_{\tilde z})$ for $\tilde z\geq z$. Hence, by monotonicity of~$R$, $\underline F(y\vert z)$ is monotone non increasing in~$z$. 

Next, if~$(y,z)\mapsto \tilde F(y\vert z)$ is a conditional cdf, such that~$\tilde F( y\vert z)\geq F(y\vert z)$ for each~$(y,z)$ and $\tilde F$ is monotone non increasing in~$z$ for each~$y$, then $\tilde F(\cdot\vert z)\geq\sup_{\tilde z\geq z}\tilde F(\cdot\vert \tilde z)\geq \sup_{\tilde z\geq z} F(\cdot\vert \tilde z)=S(K_z)$. Now, $R(\tilde F(\cdot\vert z))=\tilde F(\cdot\vert z)$ since the latter is right-continuous. Finally, monotonicity of~$R$ yields $\tilde F(\cdot\vert z)\geq R(S(K_z)):=\underline F(\cdot\vert z)$, which is the desired result.
\end{proof}

\begin{proof}[Proof of Proposition~\ref{prop:min}]
The validity of the bounds is proved in the same way as in Corollary~\ref{cor:pointwise}. We now prove~$\underline C\in\mathcal C(\pi)$. The proof of~$\bar C\in\mathcal C(\pi)$ is identical. 

(1) First, we show that $y\mapsto y-\underline{C}(y,z)$ is increasing.
Since~$\pi\in\Pi^r$, $y\mapsto \underline{F}(y|z)-\mathbb{P}(Y\leq y,D=0|z)$ is
increasing and continuous and thus, so is 
$y\mapsto L(y|z):=\sup_{\tilde y\leq y}\{\underline{F}(\tilde y|z)-\mathbb{P}(Y\leq \tilde y,D=0|z)\}$. So $y-\underline{C}(y,z)=L_{-}\left(\mathbb{P}(Y\leq y,D=1|z)|z\right)=L^{-1}\left(\mathbb{P}(Y\leq y,D=1|z)|z\right)$
is increasing in $y$ for all $z\in\mathcal{Z}$ because
both $L^{-1}(y|z)$ and $\mathbb{P}(Y\leq y,D=1|z)$ are
increasing.
(2) Second, we prove that $\underline{C}(y,z)$ is non negative. We have
$L(y|z)=\underline{F}(y|z)-\mathbb{P}(Y\leq y,D=0|z)\geq\mathbb{P}(Y\leq y|z)-\mathbb{P}(Y\leq y,D=0|z)=\mathbb{P}(Y\leq y,D=1|z).$
So we have that $L(y|z)\geq\mathbb{P}(Y\leq y,D=1|z)$ and we already
know that $L^{-1}(y|z)$ is increasing and continuous, so that
$L^{-1}\left(L(y|z)|z\right)\geq L^{-1}\left(\mathbb{P}(Y\leq y,D=1|z)|z\right)$, hence
$y\geq L_{-}\left(\mathbb{P}(Y\leq y,D=1|z)|z\right)
$ or $0\leq y-L_{-}\left(\mathbb{P}(Y\leq y,D=1|z)|z\right)=\underline{C}(y,z)$, as desired.
(3)~We now show that $\mathbb{P}(Y-\underline{C}(Y,z)\leq y,D=1|z)=\underline{F}(y|z)-\mathbb{P}(Y\leq y,D=0|z)$.
As we have shown before, $\underline{C}(y,z)=y-L^{-1}\left(\mathbb{P}(Y\leq y,D=1|z)|z\right)$.
Hence~$\mathbb{P}(Y-\underline{C}(Y,z)\leq y,D=1|z)=\mathbb P(L^{-1}(F_1(Y\vert z)\vert z)\leq y,D=1\vert z)=\mathbb P(F_1(Y\vert z)\leq L(y\vert z),D=1\vert z)=L(y,z),$ since~$F_1$ is continuous and increasing.
(4)~Finally, by Theorem~\ref{thm:char}, it suffices to show that $Y-D\underline{C}(Y,Z)$
is stochastically monotone with respect to $Z$. 
We have
\begin{eqnarray*}
\mathbb{P}(Y-D\underline{C}(Y,Z)\leq y|z) & = & \mathbb{P}(Y-D\underline{C}(Y,Z)\leq y,D=1|z)+\mathbb{P}(Y-D\underline{C}(Y,Z)\leq y,D=0|z) \\
& = & \mathbb{P}(Y-\underline{C}(Y,Z)\leq y,D=1|z)+\mathbb{P}(Y\leq y,D=0|z) \\
& = & \underline{F}(y|z)-\mathbb{P}(Y\leq y,D=0|z)+\mathbb{P}(Y\leq y,D=0|z) \; = \; \underline{F}(y|z)
\end{eqnarray*}
and $\underline{F}(y|z)$ is stochastically monotone by construction.
\end{proof}

\begin{proof}[Proof of Proposition~\ref{prop:iter}]

We show the result for the lower bound. The proof relative to the upper bounds is identical.

{\bf Step 1:} We show that~$\underline C^{(n)}$ is a valid lower bound for all~$n\geq1$. We take any~$C\in\mathcal C(\pi),$ and show that~$C\geq\underline C^{(n)}$. By Corollary~\ref{cor:pointwise}, it holds for~$n=0$. Assume it is true for~$n-1$.
In what follows,~$F_d(y\vert z):=\mathbb P(Y\leq y,D=d\vert Z=z)$, for~$d=0,1$. By  assumption, $Y-D \underline{C}^{(n-1)}(Y,Z) \geq Y-DC(Y,Z)$. For any $(y,z) \in \mathcal Y \times \mathcal Z$, we have the following sequence of implications:
\begin{eqnarray*}
\begin{array}{l}
\mathbb P(Y - D \underline{C}^{(n-1)}(Y,Z)\leq y\vert Z=z) \; \leq \; \mathbb P(Y-DC(Y,Z)\leq y\vert Z=z) \\ \\
\Rightarrow \lim_{\tilde y\downarrow y}\sup\left\{\mathbb P(Y- D \underline{C}^{(n-1)}(Y,Z) \leq \tilde y\vert \tilde z):\tilde z\in\mathcal Z,\tilde z\geq z\right\} \; \leq \; \mathbb P(Y-DC(Y,Z)\leq y\vert Z=z) \\ \\
\Rightarrow \underline F^{(n)}(y \vert z) \; \leq \; \mathbb P(Y-DC(Y,Z)\leq y\vert Z=z) \\ \\
\Rightarrow \underline F^{(n)}(y \vert z) - F_0(y \vert z) \; \leq \; \mathbb P(Y-DC(Y,Z)\leq y\vert Z=z) - F_0(y \vert z) \\ \\
\Rightarrow \underline F^{(n)}(y \vert z) - F_0(y \vert z) \; \leq \; \mathbb P(Y-C(Y,Z)\leq y,D=1\vert Z=z) \\ \\
\Rightarrow \sup_{\tilde y\leq y} \left\{{\underline F^{(n)}(y \vert z) - F_0(y \vert z)}\right\} \; \leq \; \mathbb P(Y-C(Y,Z)\leq y,D=1\vert Z=z) \\ \\
\Rightarrow L^{(n)}(y\vert z) \; \leq \; \mathbb P(Y-C(Y,Z)\leq y,D=1\vert Z=z).
\end{array}
\end{eqnarray*}
The second line uses the fact that $Y-DC(Y,Z)$ is stochastically monotone with respect to $Z$. The third line follows by the definition of~$\underline F^{(n)}(y \vert z)$. The fifth line holds by monotonicity of $\mathbb P(Y-C(Y,Z)\leq y,D=1\vert Z=z)$. We now use the above expression, to establish the desired inequality. Evaluating the function $L^{(1)}(. \vert z)$ at $y-C(y,z)$, we have:
\begin{eqnarray*}
\begin{array}{l}
L^{(n)}(y -C(y,z) \vert z) \; \leq \; \mathbb P(Y-C(Y,Z)\leq y -C(y,z),D=1\vert Z=z) \\ \\
\Rightarrow  L^{(n)}(y -C(y,z) \vert z) \; \leq \; \mathbb P(Y\leq g^{-1}(y -C(y,z)),D=1\vert Z=z) \\ \\
\Rightarrow  L^{(n)}(y -C(y,z) \vert z) \; \leq \; F_1(y \vert z) \\ \\
\Rightarrow  y -C(y,z)  \; \leq \; L_{-}^{(n)} (F_1(y \vert z)).\\ \\
\Rightarrow  C(y,z)  \; \leq \; y -L_{-}^{(n)} (F_1(y \vert z)) = \underline C^{(n)}(y,z),
\end{array}
\end{eqnarray*}
which was the desired result. This completes the proof that~$C\geq\underline C^{(n)}$ for all~$n\geq1$.

{\bf Step 2:} We now show monotonicity of the sequences, i.e., for alll~$(y,z) \in \mathcal Y \times \mathcal Z$, $ \underline F^{(n)}(y\vert z) \geq \underline F^{(n-1)}(y\vert z)$, $  L^{(n)}(y\vert z) \geq L^{(n-1)}(y\vert z)$, and $\underline C^{(n)}(y,z) \geq \underline C^{(n-1)}(y,z) \geq 0$.
Let $(y,z) \in \mathcal Y \times \mathcal Z$. Note that, by construction,
\begin{eqnarray}
\label{eq:y_dc_distr}
\begin{array}{l}
\mathbb P(Y - D \underline{C}^{(n-1)}(Y,Z)\leq y\vert Z = z)\\ \\
= P(Y - \underline{C}^{(n-1)}(Y,Z)|Z=z)\leq y ,D=1 \vert z) + \mathbb P(Y\leq y,D=0\vert Z=z),\\ \\
= P(Y - Y +  L_{-}^{(n-1)}( F_1(Y\vert z) \vert Z=z)\leq y ,D=1 \vert z) + \mathbb P(Y\leq y,D=0\vert Z=z),\\ \\
= P( F_1(Y\vert z)|z) \leq L^{(n-1)}(y \vert z),D=1 \vert z) + \mathbb P(Y\leq y,D=0\vert Z=z), \\  \\
= L^{(n-1)}(y \vert z) + F_0(y\vert z).
\end{array}
\end{eqnarray}
Similarly, we can show, $\mathbb P(Y - D \underline{C}^{(n)}(Y,Z)\leq y\vert z) = L^{(n)}(y \vert z) + F_0	(y\vert z)$.
Using the above result, we obtain:
\begin{eqnarray}
\label{eq:monotone_F}
\begin{array}{lll}
\underline F^{(n)}(y\vert z) &= & \lim_{\tilde y\downarrow y}\sup\left\{L^{(n-1)}(y \vert z) + F_0(y\vert z) :\tilde z\in\mathcal Z,\tilde z\geq z\right\},\\ \\
&= & \lim_{\tilde y\downarrow y}\sup\left\{\sup_{y' \leq \tilde y}\left\{\underline F^{(n-1)}(y'\vert z)- F_0(y'\vert z)\right\} + F_0(\tilde y\vert z) :\tilde z\in\mathcal Z,\tilde z\geq z\right\},\\ \\
&\geq & \underline F^{(n-1)}(y\vert z)
\end{array}
\end{eqnarray}
The first equality uses Equation \ref{eq:y_dc_distr}, and the second equality uses the definition of $L^{(n-1)}(y \vert z)$. The inequality in the last line uses the fact that the supremum is taken over a set that contains $F^{(n-1)}(y\vert z)$. The ordering of $\underline F^{(n-1)}(y\vert z)$ and $\underline F^{(n)}(y\vert z)$ implies easily that:
\begin{eqnarray*}
\begin{array}{lll}
 \underline F^{(n)}(y\vert z) - F_0(y \vert z) \geq \underline F^{(n-1)}(y\vert z) - F_0(y \vert z)
 &\Rightarrow &  L^{(n)}(y\vert z) \geq L^{(n-1)}(y\vert z)
 \end{array}
 \end{eqnarray*}
The ordering of $L^{(n-1)}(y\vert z)$ and $L^{(n)}(y\vert z)$ further implies that: 
\begin{eqnarray*}
\begin{array}{lll}
 L_{-}^{(n)}(y\vert z) \leq L_{-}^{(n-1)}(y\vert z) &\Rightarrow &   y- L_{-}^{(n)}(y\vert z) \geq y-L_{-}^{(n-1)}(y\vert z).
 \end{array}
 \end{eqnarray*}

Therefore, $\{\underline F^{(n)}(y\vert z): n=1,2, \ldots.\}$ is a pointwise increasing sequence defined on a compact set. ($F^{(n)}(y\vert z)$ is a cdf.) Thus, it must converge to a limit $F^{(\infty)}(y\vert z)$, which is a fixed-point of the sequence. Defining the corresponding functions $L^{(\infty)} (y\vert z)$ and $\underline{C}^{(\infty)}(y,z)$ from~$F^{(\infty)}(y\vert z)$, we have, for any $y \in \mathcal Y$ and $z \in \mathcal Z$:
\begin{eqnarray*}
\begin{array}{l}
\mathbb P(Y - D \underline{C}^{(\infty)}(Y,Z)\leq y\vert Z = z) \\ \\
= L^{(\infty)}(y \vert z) + F_0(y\vert z) , \\  \\
= \sup_{\tilde y \leq y}\left\{\underline F^{(\infty)}(\tilde y \vert z)- F_0(\tilde  y \vert z)\right\} + F_0(y\vert z) ,\\ \\
= \sup_{\tilde y \leq y}\left\{\lim_{y'\downarrow \tilde y}\sup\left\{\mathbb P(Y- D \underline{C}^{(\infty)}(Y,Z) \leq y' \vert Z = \tilde z):\tilde z\in\mathcal Z,\tilde z\geq z\right\}- F_0(\tilde y\vert z)\right\} + F_0(y\vert z) ,\\ \\
= \sup\left\{\sup_{\tilde y \leq y}\left\{\lim_{y'\downarrow \tilde y}\mathbb P(Y- D \underline{C}^{(\infty)}(Y,Z) \leq y' \vert Z = \tilde z)- F_0(\tilde y\vert z) + F_0(y\vert z)\right\}:\tilde z\in\mathcal Z,\tilde z\geq z\right\} ,\\ \\
\geq  \sup\left\{\mathbb P(Y- D \underline{C}^{(\infty)}(Y,Z) \leq y \vert Z = \tilde z):\tilde z\in\mathcal Z,\tilde z\geq z\right\}.
\end{array}
\end{eqnarray*}
The first equality applies a derivation similar to Equation \ref{eq:y_dc_distr} to obtain a relation between the cdf of $Y-DC^{(\infty)}(Y,Z)$ and $L^{(\infty)}$. The fourth equality permutes the supremum over the sets $\{\tilde y \in \mathcal Y , \tilde y \leq y\}$ and $\{\tilde z \in \mathcal Z, \tilde z \geq z\}$. The inequality in the last line uses the fact that the set $\{\tilde y \in \mathcal Y , \tilde y \leq y\}$ trivially includes $y$. The above derivation shows that $Y - D \underline{C}^{(\infty)}(Y,Z)$ is stochastic monotone with respect to $Z$. Finally, $\underline C^{(\infty)}\geq \underline C\geq0$, and~$y - \underline C^{(\infty)}(y,z) \; = \; L_{-}^{(\infty)}\left(F_1(y\vert z)\vert z\right)$ is increasing, since~$F_1(.\vert z)$ and~$L_-^{(\infty)}(.|z)$ both are. This completes the proof of~$\underline C^{(\infty)}\in\mathcal C(\pi)$.
\end{proof}

\begin{proof}[Proof of Lemma~\ref{lemma:CLR}]
Since~$\pi\in\Pi^r$ from Definition~\ref{def:adj}, $L(y\vert z)$ is continuous and increasing. Hence, we have:
\begin{eqnarray*}
L(y\vert z)&=&\sup_{\tilde y\leq y}\left\{ \underline F(\tilde y\vert z) - \mathbb P(Y\leq \tilde y,D=0\vert z)\right\} \\
&=& \underline F(\tilde y\vert z) - \mathbb P(Y\leq \tilde y,D=0\vert z) \\
&=& \lim_{\tilde y\downarrow y}\sup_{\tilde z\leq z}\mathbb P(Y\leq\tilde y\vert \tilde z)-\mathbb P(Y\leq y,D=0\vert z) \\
&=& \lim_{\tilde y\downarrow y}\left\{\sup_{\tilde z\leq z} 
 \mathbb P(Y\leq\tilde y\vert \tilde z)-\mathbb P(Y\leq \tilde y,D=0\vert z) \right\} \\
&=& \sup_{\tilde z\leq z}G(y\vert z,\tilde z).
\end{eqnarray*}
Since~$L$ is continuous and increasing, it can also be written
\begin{eqnarray*}
L(y\vert z)= \sup_{\tilde y\leq y}\sup_{\tilde z\leq z}G(\tilde y\vert z,\tilde z) 
= \sup_{\tilde z\leq z} \sup_{\tilde y\leq y}G(\tilde y\vert z,\tilde z).
\end{eqnarray*}
Defining~$\tilde G(y\vert z,\tilde z):=\sup_{\tilde y\leq y}G(\tilde y\vert z,\tilde z)$, $\tilde G(y\vert z,\tilde z)\leq\sup_{\tilde z\leq z}\tilde G(y\vert z,\tilde z)$ implies
\[
L_-(x\vert z):=\sup\{y: L(y\vert z)\leq x\}\leq \tilde G_-(x\vert z,\tilde z):=\sup\{y: \tilde G(y\vert z,\tilde z)\leq x\},
\]
for all~$\tilde z\leq z$. Hence~$L_-(x\vert z)\leq \inf_{\tilde z\leq z}\tilde G_-(x\vert z,\tilde z)$.
Finally, by continuity of~$G$ and by the definition of~$\tilde G$, we have~$\tilde G_-(x\vert z,\tilde z)\leq G_-(x\vert z,\tilde z):=\sup\{y: G(y\vert z,\tilde z)\leq x\}$. The upper bound is treated similarly. 
\end{proof}

% HM

\section{Point identification of the cost function}\label{app:HM}
In this section, we describe how Equation~(\ref{eq:maxIF}) can be used to prove the identification result of \cite{HM:2011}. The extended Roy model in this section does not rely on the restriction that the cost function associated with Sector~1 be non negative. However, it assumes that the cost function is a function of exogenous observables only (here the vector $Z$) and that the potential outcomes have a separable representation summarized in the following statement of Assumptions in \cite{HM:2011}.

\begin{assumption}\label{ass:HM}
Observable and potential outcomes are related by $Y=DY_1+(1-D)Y_0$. The selection indicator~$D$ satisfies $\mathbb E[Y_d-dC(Z)\vert\mathcal I]>\mathbb E[Y_{1-d}-(1-d)C(Z)\vert\mathcal I]\Rightarrow D=d$, where $Z$ is measurable with respect to the agent's information $\sigma$-algebra~$\mathcal I$, and $\mathbb E[Y_1-Y_0\vert \mathcal I]=\mathbb E[Y_1-Y_0\vert Z]+V$ with $V\perp Z$.
\end{assumption}

If~$C$ is a function of~$z$ only and $\mathbb E[Y_1-Y_0\vert \mathcal I]=\mathbb E[Y_1-Y_0\vert Z]+V$ with $V\perp Z$, then (\ref{eq:maxIF}) yields $\mathbb E[Y-DC(Z)\vert\mathcal I]=\max\{0,\mathbb E[Y_1-Y_0-C(Z)\vert\mathcal I]\}+\mathbb E[Y_0\vert \mathcal I]$, hence $\mathbb E[Y-Y_0\vert \mathcal I]-C(Z)\mathbb P[D=1\vert\mathcal I]=\max\{0,\mathbb E[Y_1-Y_0\vert Z]-C(Z)+V\}.$ Iterated expectations then yields 
\begin{eqnarray}\label{eq:HM1}
\mathbb E[Y-Y_0\vert Z]-C(Z)\mathbb P[D=1\vert Z]=\mathbb E[\max\{0,\mathbb E[Y_1-Y_0\vert Z]-C(Z)+V\}\vert Z].
\end{eqnarray} 
To differentiate both terms with respect to the first component of~$Z$, we make the following regularity assumptions, where $Z_{-1}$ (resp. $z_{-1}$) is the vector of components of~$Z$ (resp. $z$) excluding~$Z_1$ (resp. $z_1$).

\begin{assumption}\label{ass:reg}
For all~$z_{-1}$ on the support of~$Z_{-1}$, the functions $z_1\mapsto C(z),$ $z_1\mapsto \mathbb E[Y_d\vert D=d,Z=z],$ $d=0,1,$ and $z_1\mapsto \mathbb E[D\vert Z=z]$ are continuously differentiable on the support of~$Z_1$ conditional on~$Z_{-1}=z_{-1}$.
\end{assumption}

By Assumption~\ref{ass:reg} and the dominated convergence theorem, the partial derivative of the right-hand side of~(\ref{eq:HM1}) is
\begin{eqnarray*}
\frac{\partial}{\partial z_1}\mathbb E[\max\{0,\mathbb E[Y_1-Y_0\vert Z]-C(Z)+V\}\vert Z]=
\mathbb P[D=1\vert Z=z]\frac{\partial}{\partial z_1}\left(\mathbb E[Y_1-Y_0\vert Z]-C(Z)\right),
\end{eqnarray*}
whereas the differentiating the left-hand side yields
\begin{eqnarray*}
&&\frac{\partial}{\partial z_1}\mathbb E[Y-Y_0\vert Z=z]-C(z)\frac{\partial}{\partial z_1}\mathbb P[D=1\vert Z=z]
-\;\mathbb P[D=1\vert Z=z]\frac{\partial}{\partial z_1}C(z).
\end{eqnarray*}
Finally, we have
\begin{eqnarray*}
&&C(z)\;\frac{\partial}{\partial z_1}\mathbb P[D=1\vert Z=z]\;=\;\frac{\partial}{\partial z_1}\mathbb E[Y-Y_0\vert Z=z]-\mathbb P[D=1\vert Z=z]\;\frac{\partial}{\partial z_1}(\mathbb E[Y_1-Y_0\vert Z=z],
\end{eqnarray*}
which can be rearranged into Equation~(2.6) in \cite{HM:2011}, and which identifies the cost function if conditional means of potential outcomes are identified.

}

\end{appendix}

%     REFERENCES    

\bibliography{Roy}
\bibliographystyle{abbrvnat}

\end{document}